\documentclass[12pt]{article}
\usepackage[margin=1in]{geometry}
\usepackage{amsfonts}
\usepackage{amsmath}
\usepackage{amssymb}
\usepackage{amsthm}
\usepackage{graphicx}
\usepackage{epsfig}

\numberwithin{equation}{section}

\newtheorem{thm}{Theorem}[section]
\newtheorem{prop}{Proposition}[section]
\newtheorem{lem}{Lemma}[section]
\newtheorem{rem}{Remark}[section]

\newtheorem{defn}{Definition}[section]
\newtheorem{exm}{Example}[section]
\newtheorem{asm}{Assumption}[section]

\newcommand{\ad}{&\!\!\!\disp}

\newcommand{\barray}{\begin{array}{ll}}
\newcommand{\earray}{\end{array}}
\newcommand{\disp}{\displaystyle}
\newcommand{\bed}{\begin{displaymath}}
\newcommand{\eed}{\end{displaymath}}
\newcommand{\bea}{\bed\begin{array}{rl}}
\newcommand{\eea}{\end{array}\eed}
\newcommand{\thmref}[1]{Theorem~{\rm \ref{#1}}}
\newcommand{\lemref}[1]{Lemma~{\rm \ref{#1}}}

\newcommand{\propref}[1]{Proposition~{\rm \ref{#1}}}
\newcommand{\defnref}[1]{Definition~{\rm \ref{#1}}}

\newcommand{\asmref}[1]{Assumption~{\rm \ref{#1}}}

\def\one{{\hbox{1{\kern -0.35em}1}}}

\begin{document}

\title{Utility Maximization of an Indivisible Market with Transaction
  Costs} 

\author{Qingshuo Song\thanks{Department of Mathematics, City
    University of Hong Kong. 83 Tat Chee Avenue, Kowloon Tong, Hong
    Kong, {\tt
      song.qingshuo@cityu.edu.hk}.} \and G. Yin\thanks{Department of
    Mathematics, Wayne State University, Detroit, Michigan 48202, {\tt
      gyin@math.wayne.edu}.
    Research of this author was supported in part by the National
    Science Foundation under DMS-0907753.} \and Chao
  Zhu\thanks{Department of Mathematical Sciences, University of
    Wisconsin-Milwaukee, Milwaukee, WI 53201, {\tt zhu@uwm.edu}.}}
\maketitle

\begin{abstract}
  This work takes up the challenges of utility maximization problem
  when the market is indivisible and the transaction costs are
  included. First there is a so-called solvency
  region given by the minimum margin requirement in the problem formulation.
  Then the
  associated utility maximization is formulated as an optimal
  switching problem. The diffusion turns out to be degenerate and
  the boundary of domain is an unbounded set.
One no longer has the continuity of the value function without
posing further conditions due to the degeneracy and the dependence
of the random terminal time on the initial data. This paper
provides sufficient conditions under which the continuity of the
value function is obtained. The essence of our approach is to find
a sequence of continuous functions locally uniformly converging to
the desired value function. Thanks to continuity, the value
function can be characterized by using the notion of viscosity
solution of certain quasi-variational inequality.



  \bigskip
  \noindent{\bf Short Title.} Utility Maximization of an Indivisible
  Market with Transaction Costs

\end{abstract}

\newpage

\setlength{\baselineskip}{0.25in}
\section{Introduction}
The study of utility optimization has a long history. Utility
maximization under the setup of Black-Scholes type models can be
traced back to \cite{Mer71}. By now, it is widely understood that
in a complete market, the optimal strategy of this problem is
attainable if an investor can make infinitesimally small
adjustments of the position frequently. Recent study indicates
that market imperfections such as transaction costs and asset
indivisibility affect virtually every transaction and generate
costs, which interfere with the trades that rational individuals
would make in a complete market (see \cite{Deg05}). As was alluded
to in the above, the two main assumptions, namely zero transaction
costs and infinite divisibility of an asset, are crucial. Failure
in either of the two assumptions results in an incomplete market,
so Merton's optimal strategy becomes non-attainable.

From a practical point of view, although the technical advancement
and the on-line trading make the transaction costs not significantly
influential, the transaction costs can hardly be ignored. As for
the other assumption, it is almost evident that an asset cannot be
infinitely divisible in any practical situation.

Incorporating transaction cost in utility maximization  has
received much attentions from both researchers and practitioners
in the past few decades. In fact, there is a vast literature on
this subject; see for example,  \cite{Cha07a,
  Cha07b, Kor98,  LMP07,   MSXZ08a,
OS02,SS94, TZ},  and the references
therein. In contrast, there are relatively few works on asset
indivisibility. Two of the exceptions are \cite{RB80} and  
\cite{Wal00}. It should be noted that most existing works on asset
indivisibility have focused on discrete-time models. Our goal in
this paper is to take up the challenges in both parts. We will
 characterize the solution of the utility maximization problems of
 an
indivisible  market with transaction cost in continuous time.

To incorporate the asset indivisibility, the stock shares in the
portfolio are restricted to a finite set of integers
$\mathcal{K}$
(to be defined in \eqref{eq:cK}).
In addition,
there is a minimum maintenance margin requirement for the
investigator; the corresponding condition is termed as a
solvency region $O$
(to be defined in \eqref{eq:Omega}). The associated utility
maximization is modeled as an optimal switching problem on
degenerate diffusion in the restricted unbounded domain. It is
noted that with nondegenerate diffusion, the value function can be
shown to belongs to, for example, $ W^{1,\infty}(O) \cap
W^{2,\infty}_{{\rm loc}} (O)$  for a bounded domain  \cite{EF79},
and $W^{1,\infty}$ for a one-dimensional unbounded domain
\cite{Pha07}.

In our work, one cannot obtain the continuity of the value
function $V$ of \eqref{eq:obj} for free since the underlying
process $(X, Y, Z)$ is degenerate and the random terminal time
$\tau$ of \eqref{eq:tau} depends on the initial condition
$(t,x,y,z)$; see the counterexample in \cite[Example 4.1]{BSY09}
with the absence of optimal switching. As a result, to
characterize the value function, we use the notion of viscosity
solution for quasi-variational inequality. It turns out to be
crucial to show the continuity of the value function with some
appropriate conditions.

The continuity of the value function in a bounded domain has been
widely discussed within the framework of classical stochastic
control theory without switching costs, known as stochastic exit
problem. When the domain is bounded, a sufficient condition for
the continuity of the value function is provided in \cite[p.
205]{FS06} by using a probabilistic approach, where the continuity
was presented in terms of the drift of the underlying diffusion.
In contrast, a generalization of the continuity in \cite{BSY09}
gave a condition taking into consideration of both the drift and
diffusion coefficients. Along another line, the stochastic
exit-time control problem has been studied by using purely
analytical methods in \cite{BB95, BCI08, Kov09, LM82} under
various setups.

In the current work, we use a probabilistic approach similar to
that of \cite{BSY09} and \cite{FS06}. We focus on utility
optimization for indivisible cost with transaction costs. The
essence depends on the verification of a continuity condition. We
note that the main effort of \cite{BSY09} is to find a sequence of
continuous functions uniformly converging to the desired value
function, taking into consideration of the sample path properties
of the diffusion processes. In this procedure,  Dini's theorem
plays an essential role to obtain the uniform convergence.
However, this approach is not directly applicable to our work.
This is because the boundary of the domain $\partial O$ is
unbounded. Because of the domain being non-compact, Dini's theorem
cannot be used. Therefore, one needs asymptotic properties of the
approximating functions. Here we devise an approximation sequence
$V^\varepsilon$ (see \lemref{l-pve}), and obtain the continuity of
$V$ by local uniform estimates using $V^\varepsilon$. The details
are in \thmref{t-vcon} in what follows.

The rest of the work is arranged as follows. The precise
formulation of the problem is given in Section \ref{sec:for}.
Section \ref{sec:cont} is devoted to continuity of the value
function. Section \ref{sec:val} analyzes properties of the value
functions. In particular, we show that the value function is the
unique viscosity solution of the quasi-variational inequality
\eqref{eq:qvi} with boundary-terminal condition  \eqref{eq:btc}.
Section \ref{sec:rem} makes some further remarks to conclude the
paper. At the end, supplemental results are included in an
appendix in Section \ref{sec:app}.

\section{Problem Formulation}\label{sec:for}
Let $(\Omega, \mathcal{F}, \mathbb{P}, \mathbb{F}
)$ be a complete filtered probability space
on which is defined a standard Brownian motion $W$,
where $\mathbb{F}=\{\mathcal{F}_t\}_{t\ge 0}$.
We assume that
the filtration $\mathbb{F}$ is generated by $W$, augmented by all
the $P$-null sets as usual. For simplicity, we assume that the
financial market consists of only two assets, a bank account with
zero interest and a risky asset. Suppose that $X^{t,x}$, the price
of the risky asset, is given by
\begin{equation}
  \label{eq:stkprice}
  X(s) = x + \int_t^s b(r, X(r)) dr + \int_t^s \sigma(r, X(r)) dW(r),
\end{equation} where $x>0$ denotes the initial price.
A bank account with positive interest can be considered in the
current setup. Other than notational complexity, such a
formulation does not introduce essential difficulties as long as
the interest rates are not stochastic. Therefore, for simplicity,
we use zero risk-free interest rate in this paper. Throughout the
paper, we use the following standing assumptions. The objective
function is an expected utility with transaction costs taken into
consideration, whose precise form will be given shortly.

\begin{asm} {\rm
    \label{a-general}\quad
    \begin{enumerate}
    \item There exists a
      $C_1>0$ such that the
      drift  $b$ and  the volatility $\sigma$ satisfy
      \begin{equation}
        \label{eq:asmbs}
        b(s,0) = \sigma(s,0) = 0, \hbox{ and }
        |b(s, x_1)-b(s, x_2)| +
        |\sigma(s, x_1) - \sigma(s,x_2)| \le C_1 |x_1 -x_2|.
      \end{equation}

    \item  The transaction cost function $c: \mathbb{Z} \mapsto
      \mathbb{R}$ satisfies
      $$c(0) = 0, \quad c(z)>0 \  \forall z\neq 0, \ \hbox{ and } \
      c(z_1) + c(z_2) \ge
      c(z_1 + z_2).$$
    \item
      The risk-averse utility function $U:[0,\infty) \to [0, \infty)$
      satisfies
      \begin{equation}
        \label{eq:utility}
        U(0) = 0, \ U'(x) > 0, \  U''(x) < 0, \ \lim_{x\to \infty} U'(x)
        = 0, \ \hbox{ and } \ \lim_{x\to 0} U'(x) = \infty,
      \end{equation}
      where $U'$ and $U''$ denote the first and the second
      derivatives of $U$ with respect to $x$, respectively.
    \end{enumerate} }
\end{asm}

With condition \eqref{eq:asmbs}, the price $X(s)$ stays
nonnegative for all $s\ge t$. Note that \eqref{eq:asmbs} also
implies linear growth of the functions $b$ and $\sigma$ in the
variable $x$, and hence \eqref{eq:stkprice} has a  unique strong
solution. For a fixed time duration $[t, T]$, an investor has an
initial wealth $y$ and holds $z$ shares of stock at price $x$, and
hence $y - zx$ is the initial amount in the bank. We denote the
$i$th nonzero trading occurs at time $\tau_i$, and assume at most
one transaction occurs at each time, i.e.,
\begin{equation}
  \label{eq:stopping}
  t^-=\tau_0 < \tau_1 < \tau_2 < \cdots < \tau_N \le T, \ \hbox{ for
    some } N.
\end{equation}
We use $Z(s) = \sum_{i=0}^{N-1} Z(\tau_i) \one_{[\tau_i,
  \tau_{i+1})}(s)$ to denote the position of the risky asset
in the portfolio at time $s$, and use $\Delta Z(s) = Z(s) - Z(s^-)$
denote the amount of transaction
traded at time $s$.  Therefore, the
associated transaction cost at the $i$th transaction is  $c(\Delta
Z(\tau_i))$.

In practice, the risky asset traded
in the market is indivisible. As a result,
we restrict the investor's position in the risky asset to a
set of finite integers $\mathcal{K}$, i.e., for some positive integer
$C_2$ and $C_3$
\begin{equation}
  \label{eq:cK}
  \mathcal{K} \triangleq \{-C_2, -C_2 +1, \ldots, 0, \ldots, C_3 - 1,
  C_3\}.
\end{equation}
Then, with the initial investment $y$, the total wealth
$\{Y^{t,x,y,z,Z}(s): t\le s \le T\}$ follows
\begin{equation*}
  \begin{array}{ll}
    d Y(s) & = Z(s) b(s, X(s)) ds + Z(s) \sigma(s, X(s)) dW(s), \quad
    s\in (\tau_i, \tau_{i+1}), \\
    Y(\tau_i) & = Y(\tau_i^-) - c(\Delta Z(\tau_i)).
  \end{array}
\end{equation*}
One can rewrite the wealth process as
\begin{equation}
  \label{eq:payoff2}
  Y(s) = y + \int_t^s Z(r) b(r, X(r)) dr + \int_t^s Z(r) \sigma(r,
  X(r)) d W(r) - \sum_{\tau_i\le s} c(\Delta Z(\tau_i)).
\end{equation}

Let the minimum maintenance margin requirement for the investor's
account be $c(-Z(s))$, i.e., $Y(s) > c(-Z(s))$. The investor will
receive a margin call at $\hat \tau = \inf\{s: Y(s) \le
c(-Z(s))\}$, if $\hat \tau< T$ occurs. Under the self-financing
rule, we assume no additional capital is available, and the
investor has to clear the risky asset (zero capital remaining
after clearance) at $\hat \tau$. In other words, we define the {\it
  solvency region} as
\begin{equation}
  \label{eq:Omega}
  O = \{(x,y,z): x>0, y>c(-z), z \in \mathcal{K}\}.
\end{equation}
Thus the state process $(X(s), Y(s), Z(s))$  satisfies the state
constraint
\begin{equation}
  \label{eq:stateconstraint}
  (X(s), Y(s), Z(s)) \in O,  \quad \hbox{Lebesgue-a.e. } s\in
  [t,T], \  \mathbb{P}-a.s. \ \omega\in \Omega
  .\end{equation} In this work, $Z(s)$ is a control variable. Note
that due to the state constraint \eqref{eq:stateconstraint}, the
control $Z(s)$ belongs to a state-dependent set $Z(s^-) +
\Gamma(Y(s^-), Z(s^-))$, where $\Gamma(\cdot, \cdot)$ is a
set-valued function given by \eqref{eq:Gamma}, and $z +
\Gamma(y,z)$ is understood as a set translation.

\begin{defn} [Admissible control space]
  \label{d-conspace} {\rm
    Given $t\in [0,T]$, the set of admissible strategies, denoted as
    $\mathcal{Z}(t,x,y,z)$, is the space of $\mathbb{F}$-adapted
    processes $Z$ over $[t,T]$ such that
    \begin{enumerate}
    \item For any $s\in [t,T]$, $Z(s) \in \mathcal{K}$
      has the following form.
      For a sequence of strictly increasing stopping times,
      \eqref{eq:stopping}
      \begin{equation}
        \label{eq:Zs}
        Z(s) = \sum_{i=0}^{N-1} Z(\tau_i) \one_{[\tau_i,
          \tau_{i+1})}(s), \quad Z(t^-) = z.
      \end{equation}
    \item
      For $i\ge 1$, $\Delta Z(\tau_{i}) \in \Gamma(Y(\tau^-_i),
      Z(\tau^-_i))$, where
      \begin{equation}
        \label{eq:Gamma}
        \Gamma(y,z) = \{\tilde z \in \mathcal{K}:
        c(\tilde z - z) + c(-\tilde z) \le y, \tilde z \neq z\}.
      \end{equation}
    \end{enumerate}
  }
\end{defn}

\begin{rem}{\rm
    \defnref{d-conspace} means the investor will trade only finitely
    many times during $[t,T]$ almost surely. If not, $Y(T) = -\infty$
    almost surely
    due to $\min_{\mathcal{K}\setminus \{0\}} c(z) >0$. Also,
    \eqref{eq:Zs} implies $Z(T) = 0$, i.e.,
    the investor will always
    clear his or her
    stock position at $T$ and will hold
    only cash in the bank. Such
    an assumption is not unusual; see for example, \cite{CJP04} and
    \cite{CST07}. On the other hand, the amount trading $\Delta
    Z(\tau_{i})$  is required
    to take a value in a state-dependent set $\Gamma(Y(\tau^-_i),
    Z(\tau^-_i))$. This is the minimum requirement to keep the
    state,
    $(X(\tau_i), Y(\tau_i), Z(\tau_i))$,
    belonging to $\bar O$ (the closure of $O$) right after the transaction,
    and prevents the investor quits the market with negative
    wealth.
  }
\end{rem}

Let the stopping time $\tau$ be
\begin{equation}
  \label{eq:tau}
  \tau = \inf\{s: Y(s) \le c(-Z(s))\}\wedge T.
\end{equation}
For a given initial state $(t,x,y,z)$, the investor's goal is to
maximize the expected utility of the total wealth
\begin{equation*}
  \label{eq:functional}
  J(t,x,y, z, Z) = \mathbb{E}[ U(Y^{t,x,y,z,Z}(\tau))]
\end{equation*}
over all admissible strategy space $\mathcal{Z}(t,x,y,z)$.
Therefore, the value function of our problem is
\begin{equation}
  \label{eq:obj}
  V(t,x,y,z) = \sup_{Z \in \mathcal{Z}(t,x,y,z)} J(t,x,y,z,Z) = \sup_{Z\in
    \mathcal{Z}(t,x,y,z)} \mathbb{E}[U(Y^{t,x,y,z,Z}(\tau))].
\end{equation}

\begin{rem}
  [Discussions on assumptions] \label{r-discuss}{\rm
    There are two key assumptions in the formulation of the
    problem. One
    is the transaction cost $c(\cdot)$ being subadditive, and the other is
    that there is at most one transaction at any time, and thus the
    representation of $Z(\cdot)$ as a piecewise constant process is well
    defined. These are reasonable assumptions from the following point of
    view. Suppose there are multiple nonzero transactions occurred at time
    $s$, i.e.,
    $$\tau_i = \tau_{i+1}= \cdots = \tau_{i+m} = s \ \hbox{ for some } i,
    m\ge 1, \ \hbox{ and } t\le s \le T,$$ and the transaction cost
    $c(\cdot)$ is not necessarily subadditive.
    Denote by $\Delta Z_k$ the number of shares traded at the $k$th
    transaction. The investor is obliged to pay the total transaction
    cost $\sum_{j=0}^{m}  c (\Delta Z_{i+j})$ at time $s$. Then, we can
    always construct another function $\tilde c(\cdot)$ by
    \begin{equation*}
      \label{eq:cost1}
      \tilde c (z) = \min\Big\{ \sum_{i=1}^n c(z_i) :  z_1 + z_2
      + \dots + z_n = z \hbox{ for some } n \Big\}.
    \end{equation*}
    and such a $\tilde c(\cdot)$ turns out to be a subadditive
    function. Therefore, the multiple transactions at time $s$ can always
    be replaced by a {\it single transaction} of the amount $\sum_{j=0}^{m}
    \Delta Z_{i+j}$ shares in terms of
    the new subadditive transaction cost
    function
    $\tilde c(\cdot)$. As a result, the strategy remains the same as
    before, while the transaction cost becomes less under $\tilde
    c(\cdot)$, i.e.,  $\sum_{j=0}^{m} c (\Delta Z_{i+j}) \ge   \tilde c
    (\sum_{j=0}^{m} \Delta Z_{i+j} )$;  see \cite{MSXZ08a} for a more
    general discussion.
  }
\end{rem}

We define two operators
\begin{equation}
  \label{eq:opL}
  \mathcal{L} \varphi(t,x,y,z)
  = b \varphi_x + \frac 1 2 \sigma^2
  \varphi_{xx} + z b \varphi_y
  + \frac 1 2 z^2 \sigma^2 \varphi_{yy} +
  z \sigma^2 \varphi_{xy},
\end{equation}
and
\begin{equation}
  \label{eq:opS}
  \mathcal{S} \varphi(t,x,y,z) = \max_{\tilde z \in \Gamma(y,z)}
  \varphi(t,x,y-c(\tilde z -z), \tilde z).
\end{equation}
In the above, we used $\displaystyle \max_{z\in \emptyset}
\varphi(t,x,y-c(\tilde z -z), \tilde z) = -\infty = -\min_{z\in
  \emptyset}  \varphi(t,x,y-c(\tilde z -z), \tilde z)$ as convention.
In the definition of $\mathcal{S}$, we  used $\max$ instead of
$\sup$ owing to the finite cardinality of $\Gamma(y,z)$.
Also note that, the operator $\mathcal{L}$ of \eqref{eq:opL} is
degenerate. In other words, the diffusion $(X, Y, Z)$ is always
degenerate, even if $X$ is non-degenerate.

Provided that $V$ is a continuous function, we can proceed with 
the dynamic programming principle (DPP) and obtain
\begin{equation*}
  \label{eq:dpp}
  V(t,x,y,z) = \sup_{Z\in \mathcal{Z}(t,x,y,z)}  \mathbb{E} [V(\theta,
  X^{t,x}(\theta), Y^{t,x,y,z,Z}(\theta), Z(\theta))], \ \forall \theta
  \le \tau.
\end{equation*}
The general discussions of DPP is referred to \cite{FS06, Pha09}.
If we appeal DPP on instantaneous
transaction strategy with $\tau_1= t$, then it follows that
\begin{equation}
  \label{eq:dpp1}
  V(t,x,y,z) \ge \mathcal{S}V (t,x,y,z), \ (t,x,y,z)\in [0,T]\times
  O.
\end{equation}
Define an operator $\mathcal{A}$
that maps from measurable functions
$\varphi: (0,T) \times O \to \mathbb{R}$ to set-valued functions
$\mathcal{A}[\phi]$ on $\mathcal{K}$ given by
\begin{equation}
  \label{eq:nonaction}
  \mathcal{A}[\varphi] (z) = \{(t,x,y): \varphi(t,x,y,z) >
  \mathcal{S} \varphi(t,x,y,z) \}.
\end{equation}
Note that $\mathcal{A}[V](z)$ is a {\it no-action region}
associated with $z\in \mathcal{K}$. DPP implies that for the
initial data $(t,x,y) \in \mathcal{A}[V](z)$,  the value process
$V(s, X^{t,x}(s), Y^{t,x,y,z}(s), z)$ is a martingale in
$A[V](z)$, whose generator is given by
$\frac{\partial}{\partial t} + \mathcal{L}$. Moreover, a heuristic
derivation leads to that $V$ satisfies the following
quasi-variational inequality
\begin{equation}
  \label{eq:qvi}
  \min\{-u_t - \mathcal{L} u, u - \mathcal{S}u\} = 0,  \hbox{ on }
  [0,T) \times O.
\end{equation}

We aim to show the value function $V$ is the unique viscosity
solution of the quasi-variational inequality  \eqref{eq:qvi} with
Cauchy-Dirichlet data
\begin{equation}
  \label{eq:btc}
  u(t,x,y,z) = U(y-c(-z)),  \hbox{ on } \partial^* ([0,T)\times O),
\end{equation}
where $\partial^* ([0,T) \times O)$ is the parabolic boundary. It
turns out to be crucial to know the continuity of $V$  {\it a
  priori}.

For later use in the uniqueness proof, we define
the function $F$ as
\begin{equation}
  \label{eq:funF}
  \begin{array}{ll}
    \displaystyle F(t,x,y,z,q,p,A)  & =  - q - (b(t,x) p_1 + \frac 1 2
    \sigma^2(t,x) A_{11}
    + z b(t,x) p_2 \\
    & \ \ \ + \frac 1 2 z^2 \sigma^2(t,x)
    A_{22} + z \sigma^2(t,x,) A_{22} + z \sigma^2(t,x) A_{12})
    . \end{array}
\end{equation}
Then, \eqref{eq:qvi} can be rewritten as
\begin{equation}
  \label{eq:qviF}
  \min\{F(t,x,y,z,u_t, Du, D^2 u), u - \mathcal{S}u\} = 0.
\end{equation}

\section{Continuity}\label{sec:cont}
Continuity is crucial to characterize the value function as the
unique viscosity solution. The difficulty to show the continuity
of $V(\cdot)$ stems from the following:
\begin{enumerate}
\item the stopping time $\tau$ of \eqref{eq:tau} depends on the
  initial
  state $(x,y)$;
\item the boundary $\partial^* ([0,T) \times O)$ is an unbounded
  set;
\item the control space $\mathcal{Z}(t,x,y,z)$ depends on the
  initial
  state $(x,y)$.
\end{enumerate}

To prove the continuity of $V(\cdot)$, we introduce another value
function $V^\varepsilon(\cdot)$ in what follows, which avoids the
above two issues of $V(\cdot)$. Let the strategy space
$\mathcal{Z}(t,z)$ be defined as a strategy space without
constraint \eqref{eq:stateconstraint}, so that the space does not
depend on the initial state $(x,y)$, i.e.,
\begin{equation*}
  \label{eq:conspc1}
  \begin{array}{l}
    \mathcal{Z}(t, z) = \{Z: Z(t^-) = z, \mathcal{K}\ni Z(s) =
    \sum_{i=0}^{N-1} Z(\tau_i) \one_{[\tau_i, \tau_{i+1})}(s) \hbox{
      for some } N,  Z(T) = 0\}.
  \end{array}
\end{equation*}
Recall that  $\tau$ of \eqref{eq:tau} is defined as
the first exit time of the
random process $(t,X^{t,x,}, Y^{t,x,y,z}, Z)$ from the domain
$[0,T)\times O$. Thus,
one can rewrite  $V$ of \eqref{eq:obj} as,
\begin{equation*}
  \label{eq:vve}
  V(t,x,y,z) = \sup_{Z\in \mathcal{Z}(t,z)} \mathbb{E}[
  U(Y^{t,x,y,z,Z}(T)) \one_{\{\tau = T\}}].
\end{equation*}
We also define $\Lambda^\varepsilon$
as a penalty function of the form
\begin{equation}
  \label{eq:penalty}
  \Lambda^\varepsilon(t,s,Y,Z) = \exp\Big\{-\frac 1 \varepsilon \int_t^s
  \Big(c(-Z(r)) - Y(r)\Big)^+ dr\Big\},
\end{equation}
where $c(z)^+$ denotes the positive part of $c(z)$ as usual.
Finally, we define $V^\varepsilon$ as
\begin{equation}
  \label{eq:obje}
  V^\varepsilon(t,x,y,z) = \sup_{Z \in \mathcal{Z}(t,z)} \mathbb{E} [
  \Lambda^\varepsilon (t,T, Y^{t,x,y,z,Z}, Z)
  U(Y^{t,x,y,z,Z} (T))].
\end{equation}
In the above, we extend the function $U(\cdot)$ to $(-\infty,
\infty)$ by $U(x) = 0$ for any $x<0$. Since $\Lambda^\varepsilon
\equiv 1$ on the set $\{\tau=T\}$, it leads to
\begin{equation*}
  \label{eq:vve-1}
  V^\varepsilon(t,x,y,z) \ge V(t,x,y,z), \ \forall (t,x,y,z).
\end{equation*}

The $V^\varepsilon(t,x,y,z)$ can be thought of as a penalized or
regularized ``value function.'' We use $V^\varepsilon$ to
establish the desired properties of $V$. The tasks to be performed
are:
\begin{enumerate}
\item to show that $V^\varepsilon(\cdot, \cdot, \cdot,z)$ is
  continuous
  for each $\varepsilon$;
\item to show that $V^\varepsilon$ converges monotonically to $V$
  in
  $[0,T)\times O$; and
\item to show that $V^\varepsilon$ converges locally uniformly to $V$.
\end{enumerate}

\subsection{Preliminary Results}
\begin{prop}
  [Properties of $\mathcal{S}$]\label{p-opS}  The following
  properties hold for the operator $\mathcal{S}$:
  \begin{enumerate}
  \item {\rm (Monotonicity)} $\mathcal{S} u \ge \mathcal{S} v$
    whenever $u\ge v$.
  \item
    {\rm (sub-distributivity)} $\mathcal{S} (u + v) \le \mathcal{S} u +
    \mathcal{S} v$.

  \item {\rm (Preservation of continuity)} $\mathcal{S} u$ is
    continuous in $(t,x,y)$ whenever $u$ is continuous in $(t,x,y)$.
  \end{enumerate}
\end{prop}
\begin{proof}
  \begin{enumerate}
  \item (Monotonicity) If $u\ge v$, then by
    definition
    \eqref{eq:opS}
    \begin{equation*}
      \begin{array}{ll}
        &\mathcal{S}u (t,x,y,z) - \mathcal{S}v (t,x,y,z) \\
        &\quad =
        \max_{\tilde z        \in \Gamma(y,z)} u(t,x,y-c(\tilde z -z),
        \tilde
        z) - \max_{\tilde z \in \Gamma(y,z)} v(t,x,y-c(\tilde z
        -z), \tilde z)   \\
        &\quad = \max_{\tilde
          z        \in \Gamma(y,z)} u(t,x,y-c(\tilde z -z), \tilde
        z) + \min_{\tilde z \in \Gamma(y,z)} (-v)(t,x,y-c(\tilde z
        -z), \tilde z)  \\
        &\quad \ge \min_{\tilde z \in \Gamma(y,z)} (u-v)(t,x,y-c(\tilde z
        -z), \tilde z) \ge 0.
      \end{array}
    \end{equation*}
  \item (sub-distributivity)
    The proof is obvious and thus omitted.
  \item (Preservation of continuity) For each pair $(z,\tilde z)$,
    $u(t,x,y-c(\tilde z - z), \tilde z)$ is continuous in
    $(t,x,y)$. Also, note that $\Gamma(y,z)$ is
    a finite set. Thus,
    $\max_{\tilde z \in \Gamma(y,z)} u(t,x,y-c(\tilde z
    -z), \tilde z) $  is also continuous.
  \end{enumerate}
\end{proof}

\begin{lem}
  \label{l-y} Let $Z\in \mathcal{Z}(t,z)$. For any $m\ge 1$,
  the wealth process $Y$ given by
  \eqref{eq:payoff2} satisfies
  \begin{equation}
    \label{eq:l-y1}
    \mathbb{E}\Big[ \sup_{t\le s\le T} |Y^{t,x,y,z,Z}(s) - y +
    \sum_{\tau_i \le s}
    c(\Delta Z_{\tau_i}) |^m \Big] \le C_{m,T} |x|^m ,
  \end{equation}
  and
  \begin{equation}
    \label{eq:l-y2}
    \mathbb{E} \Big[\sup_{t\le s\le T} |Y^{t,x_1,y_1,z,Z}(s) -
    Y^{t,x_2,y_2,z,Z}(s)|^m \Big] \le C_{m,T} (|x_1 - x_2|^m + |y_1 -
    y_2|^m).
  \end{equation}
\end{lem}

\begin{proof}
  We denote
  $Y^Z \triangleq Y^{t,x,y,z,Z}$ and $X \triangleq
  X^{t,x}$. Using the Burkholder-Davis-Gundy (BDG)
  and H\"older inequalities
  multiple  times combined with linear growth and Lipschitz
  conditions in \eqref{eq:asmbs}, we compute
  \begin{equation*}
    \begin{array}{l}
      \mathbb{E}\Big[ (\sup_{t\le s\le T} |Y^{t,x,y,z,Z}(s) - y +
      \sum_{\tau_i \le s}   c(\Delta Z_{\tau_i}) |^m \Big] \\
      \displaystyle \le C_m
      \mathbb{E}\Big[
      \sup_s
      \Big|\int_t^s Z(r) b(r,X(r)) dr \Big|^m \Big]
      + C_m \mathbb{E} \Big[ \sup_s \Big|\int_t^s Z(r) \sigma
      (r,X(r))  dW(r)\Big|^m \Big]\\ \displaystyle
      \le C_m
      \mathbb{E}\Big[
      \sup_s
      \int_t^s |Z(r) b(r,X(r))|^m dr  \Big]
      + C_m \mathbb{E} \Big[  \Big(\int_t^T Z^2(r) \sigma^2(r,X(r))
      dr\Big)^{m/2}     \Big]       \\ \displaystyle
      \le  C_m
      \mathbb{E}\Big[ \int_t^T |b(r,X(r))|^m dr  \Big]
      + C_m \mathbb{E} \Big[  \Big(\int_t^T \sigma^2(r,X(r))
      dr\Big)^{m/2}        \Big]       \\
      \displaystyle        \le   C_m
      \mathbb{E}\Big[ \int_t^T |X(r)|^m dr  \Big]
      + C_m \mathbb{E} \Big[  \Big(\int_t^T |X(r)|^2
      dr\Big)^{m/2}        \Big]       \\
      \le  C_{m,T} |x|^m.
    \end{array}
  \end{equation*}
  Then  \eqref{eq:l-y1} follows.
  For convenience, we also denote
  $Y^{i,Z} \triangleq Y^{t,x_i,y_i,z,Z}$  and
  $X^i  \triangleq  X^{t,x_i}$ for $i = 1,2$. Similar arguments
  lead to
  \begin{equation*}
    \begin{array}{l}
      \mathbb{E} \Big[ \sup_{t\le s\le T} |Y^{1,Z}(s) - Y^{2,Z}(s)|^m
      \Big] \\
      \displaystyle \le C_m |y_1 - y_2|^m + C_m \mathbb{E} \Big[\sup_s
      \Big| \int_t^s Z(r) (b(r, X^1(r)) - b(r, X^2(r))) dr \Big|^m \Big] \\
      \displaystyle \hspace{1in} + C_m
      \mathbb{E} \Big[ \sup_s \Big| \int_t^s Z(r) (\sigma(r,X^1(r)) -
      \sigma (r,X^2(r))) dW(r)\Big|^m \Big] \\
      \displaystyle \le C_m |y_1 - y_2|^m + C_m \mathbb{E} \Big[\sup_s
      \int_t^s |Z(r) (b(r, X^1(r)) - b(r, X^2(r))) |^m dr \Big] \\
      \displaystyle \hspace{1in} + C_m
      \mathbb{E} \Big[ \Big(  \int_t^T Z^2(r) (\sigma(r,X^1(r)) -
      \sigma (r,X^2(r)))^2  dr \Big)^{m/2} \Big] \\
      \displaystyle \le C_m |y_1 - y_2|^m + C_m \mathbb{E} \Big[
      \int_t^T |X^1(r) - X^2(r)|^m dr \Big] + C_m
      \mathbb{E} \Big(  \int_t^T |X^1(r) -
      X^2(r)|^2  dr \Big)^{m/2} \\
      \le  C_m |y_1 - y_2|^m +  C_{m,T} |x_1 -x_2|^m.
    \end{array}
  \end{equation*}
\end{proof}

\subsection{Properties of $V^\varepsilon$}
\begin{lem}
  \label{l-pve}
  $V^\varepsilon(t,x,y,z)$ is increasing in $y$, and continuous in
  $(t,x,y)$. Furthermore, $V^\varepsilon$ satisfies
  \begin{equation}
    \label{eq:l-pve1}
    \lim_{x\to \infty} \sup_{t,\varepsilon} \frac{V^\varepsilon(t,x,y,z)}{x} = 0,
    \ \forall
    (y,z),
  \end{equation}
  and
  \begin{equation}
    \label{eq:l-pve1-1}
    \lim_{y\to \infty} \sup_{t,\varepsilon} \frac{V^\varepsilon(t,x,y,z)}{y} = 0,
    \ \forall
    (x,z),
  \end{equation}
\end{lem}

\begin{proof}
  It is clear
  that $V^\varepsilon$ is increasing in $y$.
  \begin{enumerate}
  \item In this part, we prove $V^\varepsilon$ is continuous in $(x,
    y)$. Given that $(t,z) \in [0,T]$ and $(x_i, y_i) \in \mathbb{R}^+
    \times \mathbb{R}$ with $i= 1,2$, we denote
    \begin{equation*}
      Y^{i,Z,+} \triangleq
      \max\{Y^{i,Z}, 0\},  \quad i = 1,2.
    \end{equation*}
    Then we have
    \begin{equation}\label{eq:l-pve2}
      \begin{array}{l}
        |V^\varepsilon(t,x_1,y_1,z) - V^\varepsilon(t,x_2,y_2,z)| \\
        =  \Big|\displaystyle\sup_{Z\in \mathcal{Z}(t,z)} \mathbb{E}[
        \Lambda^\varepsilon (t, T, Y^{1,Z}, Z) U( Y^{1,Z,+}(T))] \\
        \hspace{1in} -
        \displaystyle \sup_{Z\in \mathcal{Z}(t,z)}
        \mathbb{E}[\Lambda^\varepsilon
        (t, T, Y^{2,Z}, Z)  U(Y^{2,Z,+}(T))] \Big| \\
        =  \displaystyle \sup_{Z\in \mathcal{Z}(t,z)} \mathbb{E} \Big|
        \Lambda^\varepsilon (t, T, Y^{1,Z}, Z) U( Y^{1,Z,+}(T)) -
        \Lambda^\varepsilon (t, T, Y^{2,Z}, Z)  U(Y^{2,Z,+}(T)) \Big|
        \\ \le \displaystyle \sup_{Z\in \mathcal{Z}(t,z)} \mathbb{E}
        \Big|
        (\Lambda^\varepsilon (t, T, Y^{1,Z}, Z) - \Lambda^\varepsilon
        (t, T, Y^{2,Z}, Z))  U( Y^{1,Z,+}(T)) \Big| \\ \hspace{1in}
        \displaystyle +
        \sup_{Z\in \mathcal{Z}(t,z)} \mathbb{E} \Big|
        \Lambda^\varepsilon (t, T, Y^{2,Z}, Z)  ( U( Y^{1,Z,+}(T)) -
        U(Y^{2,Z,+}(T))) \Big| \\
        \le \displaystyle \sup_{Z\in \mathcal{Z}(t,z)} \|
        \Lambda^\varepsilon (t, T, Y^{1,Z}, Z) - \Lambda^\varepsilon
        (t, T, Y^{2,Z}, Z)\|_2  \|U( Y^{1,Z,+}(T))\|_2 \\ \hspace{2in}
        \displaystyle +
        \sup_{Z\in \mathcal{Z}(t,z)} \mathbb{E} \Big|
        ( U( Y^{1,Z,+}(T)) - U(Y^{2,Z,+}(T))) \Big|.
      \end{array}
    \end{equation}
    The last inequality of \eqref{eq:l-pve2}
    follows from H\"older's
    inequality and
    the fact $|\Lambda^\varepsilon|\le 1$.
    In the above and what follows,
    we use $\|\cdot\|_2$ to denote the
    norm in the space  $L^2(\Omega, \mathcal{F}, \mathbb{P};
    \mathbb{R})$.
    To proceed, we
    examine each
    of the terms
    after the last inequality sign in
    \eqref{eq:l-pve2}.

    Since we have
    \begin{equation}
      \label{eq:l-pve4}
      \begin{array}{l}
        |\Lambda^\varepsilon (t,T, Y^{1,Z}, Z)  -
        \Lambda^\varepsilon(t,T, Y^{2,Z}, Z) |  \\
        = \displaystyle \Big| \exp\{-\frac 1 \varepsilon \int_t^T
        (c(-Z(r)) - Y^{1,Z}(r))^+ dr\}  \\ \hspace{1in} \displaystyle
        - \exp\{-\frac 1 \varepsilon \int_t^T
        (c(-Z(r)) - Y^{2,Z}(r))^+ dr\} \Big| \\ \displaystyle
        \le \frac 1 \varepsilon \Big|\int_t^T \Big((c(-Z(r)) -
        Y^{1,Z}(r))^+
        - (c(-Z(r)) - Y^{2,Z}(r))^+ \Big) dr \Big|
        \\ \displaystyle \le \frac 1 \varepsilon \int_t^T |Y^{1,Z}(r)
        - Y^{2,Z}(r)| dr \\ \displaystyle
        \le C_{\varepsilon,T} \sup_{r\in [t,T]}  |Y^{1,Z}(r)
        - Y^{2,Z}(r)| \ \hbox{ a.s.,}
      \end{array}
    \end{equation}
    the
    first factor 
    in the next to the last row of \eqref{eq:l-pve2} is
    \begin{equation}
      \label{eq:l-pve5}
      \begin{array}{ll}
        \|\Lambda^\varepsilon (t,T, Y^{1,Z}, Z)  -
        \Lambda^\varepsilon(t,T, Y^{2,Z}, Z) \|_2  \\
        = \Big( \mathbb{E} | \Lambda^\varepsilon (t,T, Y^{1,Z}, Z)  -
        \Lambda^\varepsilon(t,T, Y^{2,Z}, Z) |^2 \Big)^{1/2} \\
        \le C_{\varepsilon,T} \Big( \mathbb{E} \sup_{r\in [t,T]}  |Y^{1,Z}(r)
        - Y^{2,Z}(r)|^2 \Big)^{1/2} \ & 
        \\  \le C_{\varepsilon,T} (|x_1 - x_2| + |y_1 - y_2|), \ &
      \end{array}
    \end{equation}
    where the last inequality follows from \eqref{eq:l-y2}.

    For the
    second factor 
    in the next to the last row
    in \eqref{eq:l-pve2}, we
    utilize the fact $U^2(x) \le C(1+x^2)$ for some large $C$ due to
    concavity of $U$
    \begin{equation}
      \label{eq:l-pve11}
      \|U( Y^{1,Z,+}(T))\|_2 = \Big(\mathbb{E} [ U^2(
      Y^{1,Z,+}(T))]\Big)^{1/2} \le C + C \| Y^{1,Z,+}(T) \|_2.
    \end{equation}
    Note that
    $|Y^{1,Z,+}(T) | \le |Y^{1,Z}(T) + \sum_{\tau_i \le T} c(\Delta
    Z_{\tau_i}) |$, one can use the result of \eqref{eq:l-y1} to obtain
    \begin{equation}
      \label{eq:l-pve6}
      \|Y^{1,Z,+}(T)\|_2 \le (\mathbb{E}  |Y^{t,x_1,y_1,Z}(s)  +
      \sum_{\tau_i \le s}  c(\Delta Z_{\tau_i}) |^2)^{1/2} \le C_{T}
      (|x_1| + |y_1|).
    \end{equation}
    Combining the inequalities \eqref{eq:l-pve11} and
    \eqref{eq:l-pve6}, we have
    \begin{equation}
      \label{eq:l-pve12}
      \|U( Y^{1,Z,+}(T))\|_2 \le C_{T}
      (|x_1| + |y_1| + 1).
    \end{equation}

    For the last term of \eqref{eq:l-pve2}, we use
    $|U(x_1) - U(x_2)| \le  U(|x_1 - x_2|)$ and Jensen's inequality to
    obtain
    \begin{equation*}
      \label{eq:l-pve7}
      \mathbb{E}  \Big|
      ( U( Y^{1,Z,+}(T)) - U(Y^{2,Z,+}(T))) \Big| \le
      U (\mathbb{E}            |Y^{1,Z,+}(T) - Y^{2,Z,+}(T)| ).
    \end{equation*}
    Also, thanks to \eqref{eq:l-y2}, we further obtain
    \begin{equation}
      \label{eq:l-pve8}
      \mathbb{E}  \Big|
      ( U( Y^{1,Z,+}(T)) - U(Y^{2,Z,+}(T))) \Big| \le
      U ( C_T (|x_1 - x_2| + |y_1 -
      y_2|)) .
    \end{equation}
    Coming back
    to \eqref{eq:l-pve2} with the estimates
    \eqref{eq:l-pve5}, 
    \eqref{eq:l-pve12}, and
    \eqref{eq:l-pve8}, we have
    \begin{equation}
      \label{eq:l-pve9}
      \begin{array}{l}
        |V^\varepsilon(t,x_1,y_1,z) -
        V^\varepsilon(t,x_2,y_2,z)|  \\
        \le
        C_{\varepsilon, T} (|x_1 - x_2| + |y_1 - y_2|) (|x_1| + |y_1|
        + 1) + U(C_T (|x_1 - x_2| + |y_1 - y_2|)).
      \end{array}
    \end{equation}
    Therefore, $V^\varepsilon$ is continuous in $(x,y)$.
  \item
    With  the continuity of $V^\varepsilon$ in $(x,y)$, we
    are now ready to establish the
    continuity of $V^\varepsilon$  in $t$. We 
    assume $t_1<t_2$ and fix $(x,y)$.  By the
    definition of    $V^\varepsilon$ in \eqref{eq:obje}, for any $Z_1\in
    \mathcal{Z}(t_1,z)$
    \begin{equation}
      \label{eq:l-pve13}
      \begin{array}{l}
        V^\varepsilon(t_1,x,y,z) - V^\varepsilon(t_2,x,y,z)    \ge \\
        \displaystyle
        \mathbb{E}_{t_1} [
        \Lambda^\varepsilon (t_1, t_2, Y^{t,x,y,z,Z_1}, Z_1)
        V^\varepsilon (t_2, X^{t_1,x}(t_2), Y^{t_1,x,y,z,Z_1}(t_2),
        Z_1(t_2))]
        - V^\varepsilon(t_2,x,y,z).
      \end{array}
    \end{equation}
    If we restrict $\sup$ of \eqref{eq:l-pve13} in $Z_1 \in
    \mathcal{Z}(t_1,z): Z_1(s) = z \ \forall s\in [t_1,t_2]$, then it
    gives a one-sided estimate
    \begin{equation}
      \label{eq:l-pve14}
      \begin{array}{l}
        V^\varepsilon(t_1,x,y,z) - V^\varepsilon(t_2,x,y,z)   \\
        \displaystyle \quad \ge
        \mathbb{E}_{t_1} [
        \Lambda^\varepsilon (t_1, t_2, Y^{t,x,y,z}, z)
        V^\varepsilon (t_2, X^{t_1,x}(t_2), Y^{t_1,x,y,z}(t_2), z)]
        - V^\varepsilon(t_2,x,y,z) \\
        \displaystyle\quad = \mathbb{E}_{t_1}
        [\Lambda^\varepsilon (t_1, t_2, Y^{t,x,y,z}, z)
        (V^\varepsilon(t_2,X^1(t_2),Y^1(t_2) ,z) -
        V^\varepsilon(t_2,x,y,z)) ] \\
        \hspace{2.5in} - \mathbb{E}_{t_1} [ (1 -
        \Lambda^\varepsilon(t_1, t_2, Y^1, z)) V^\varepsilon(t_2, x,
        y, z)].
      \end{array}
    \end{equation}
    The last
    term of \eqref{eq:l-pve14} vanishes
    as $t_2
    \to t_1$ by
    the dominated convergence theorem.
    The term on the next to the last line
    also goes to zero as $t_2 \to t_1$,
    due to
    \begin{enumerate}
    \item estimation of \eqref{eq:l-pve9} on  $V^\varepsilon$ in $(x,y)$
    \item
      the inequality
      \begin{equation*}
        \mathbb{E}[\sup_{t_1 \le s t_2} (|X^1(t_2) - x|^m + |Y^1(t_2)
        -y|^m] \le C_m (1 + |x|^m) (t_2 - t_1)^{m/2}; \hbox{ and }
      \end{equation*}
    \item
      $|\Lambda^\varepsilon| \le 1$.
    \end{enumerate}
    Therefore, $\lim_{t_2 \to t_1}        (V^\varepsilon(t_1,x,y,z) -
    V^\varepsilon(t_2,x,y,z)) \ge 0$, and $V^\varepsilon$ is left
    upper semicontinuous.
    For any $Z \in Z(t_1,z)$, we design $\hat Z (s) = Z(s)$
    for all $s\ge t_2$, and $\hat Z(t_2^-) = z$. Then $\hat Z \in
    \mathcal{Z}(t_2,z)$. Thus,
    \begin{equation*}
      \label{eq:l-pve15}
      \begin{array}{l}
        V^\varepsilon(t_1,x,y,z) - V^\varepsilon(t_2,x,y,z)   \\ \le
        \displaystyle
        \sup_{Z\in \mathcal{Z}(t_1,z)} \Big\{ \mathbb{E}_{t_1}
        [\Lambda^\varepsilon (t_1, t_2, Y^{1,Z}, Z)
        V^\varepsilon(t_2, X^1(t_2), Y^{1,Z}(t_2), Z(t_2))] \\
        \quad \hspace{3in}  \displaystyle - \mathbb{E}_{t_1} [
        J^\varepsilon(t_2, x,
        y, z, \hat Z)] \Big\}\\ \displaystyle \le
        \sup_{Z\in \mathcal{Z}(t_1,z)} \Big \{\mathbb{E}_{t_1}
        [\Lambda^\varepsilon (t_1, t_2, Y^{1,Z}, Z)
        V^\varepsilon(t_2, X^1(t_2), Y^{1,Z}(t_2), Z(t_2))] \\
        \quad \hspace{2.5in} \displaystyle - \mathbb{E}_{t_1} [
        V^\varepsilon(t_2, x,
        y- c(Z(t_2)        -z), Z(t_2))] \Big\}\\ \displaystyle \le
        \sup_{Z\in \mathcal{Z}(t_1,z)} \mathbb{E}_{t_1} \Big[
        V^\varepsilon(t_2, X^1(t_2), Y^{1,Z}(t_2), Z(t_2))  -
        V^\varepsilon(t_2, x, y- c(Z(t_2)-z), Z(t_2)) \Big].
      \end{array}
    \end{equation*}
    Observe that by the sub-additivity of $c(\cdot)$,
    \bea
    Y^{1,Z}(t_2)\ad = y + \int_{t_1}^{t_2} Z(s) dX(s) -
    \sum_{\tau_i \le t_2} c(\Delta Z(\tau_i))\\
    \ad \le y +
    \int_{t_1}^{t_2} Z(s) dX(s) - c(Z(t_2) - z).
    \eea
    Together with monotonicity of $V^\varepsilon$ in $y$, we obtain
    the
    desired estimate $$\lim_{t_2 \to t_1}
    (V^\varepsilon(t_1,x,y,z) -
    V^\varepsilon(t_2,x,y,z)) \le 0.$$
    In other words, $V^\varepsilon$ is left lower semicontinuous in
    $t$. Right continuity can be similarly shown along the above lines
    by forcing the limit $t_1\to t_2$.
  \item
    Note that by virtue of \eqref{eq:l-y1},
    \begin{equation*}
      \begin{array}{ll}
        V^\varepsilon(t,x,y,z)  & \le \sup_{Z\in \mathcal{Z}(t,z)}
        \mathbb{E} [U(Y^{t,x,y,z,Z}(T))] \\
        & \le  \sup_{Z\in \mathcal{Z}(t,z)} U(
        \mathbb{E}[Y^{t,x,y,z,Z}(T)]) \\
        & \le  \sup_{Z\in \mathcal{Z}(t,z)} U( y + C x).
      \end{array}
    \end{equation*}
    This, together with \eqref{eq:utility}, implies \eqref{eq:l-pve1}
    and \eqref{eq:l-pve1-1}.
  \end{enumerate}
\end{proof}

\subsection{Continuity of $V$}
\begin{asm}
  \label{a-con}
  For any $(t,x)\in (0,T) \times \mathbb{R}^+ $ and $0 \neq z\in
  \mathcal{K}$, either $zb(t,x) <0$ or $\sigma(t,x) \neq 0.$
\end{asm}

\begin{rem}{\rm
    If $\mathcal{K}$ includes both negative and positive integers, then
    $z b(t,x)<0$ is meaningless. But  if $\mathcal{K}$
    only contains
    nonnegative integers
    (that is, short position is prohibited),
    then $zb(t,x)<0$
    leads to $b(t,x)<0$.}
\end{rem}

Define the effective boundary of the domain as follows:
\begin{equation}
  \label{eq:p1Omega}
  \partial^1 O = \{(x,y,z): x>0, y = c(-z), z \in \mathcal{K}\}.
\end{equation}

\begin{lem}
  \label{l-as}
  For arbitrarily given initial data $(t,x,y,z) \in [0,T) \times
  \partial^1 O \cap \{z\neq
  0\}$ and $Z\in \mathcal{Z}(t,z)$, let $Y \triangleq Y^{t,x,y,z,Z}$ be
  a process of \eqref{eq:payoff2}. Under \asmref{a-con},  we have
  \begin{equation*}
    \label{eq:l-as1}
    \inf\{s>t: Y(s) < C(-Z(s)) \} = t \quad \mathbb{P}-a.s.
  \end{equation*}
\end{lem}
\begin{proof}
  Given $Z \in \mathcal{Z}(t,z)$, we define $A = \{\omega: Z(t,\omega)
  = z\}$. For any $\omega \notin A$, one can see
  $$Y(t,\omega) = c(-z) - c(Z(t,\omega) - z) < c(-z),$$
  and thus,
  $$\inf\{s>t: Y(s)< c(-Z(s))\} = t \quad \mathbb{P}-a.s. \hbox{ in }
  \Omega\setminus A.$$
  Next, we want to show
  $$\inf\{s>t: Y(s)< c(-Z(s))\} = t \quad \mathbb{P}-a.s. \hbox{ in }
  A.$$
  Let $\rho(y,z) = c(-z) - y$. Consider $Z^1\in \mathcal{Z}(t,z)$
  given by
  $$Z^1(s, \omega) = Z(s, \omega) \one_{A} (\omega) + z \one_{A^c}
  (\omega), \ \forall s\in [t, T).$$
  In other words, $Z^1$ is constructed
  so that if
  there is a jump at $t$, then $Z^1$ follows exactly
  the sample path as $Z$,
  and if not $Z^1$ just takes constant $z$ before clear all risky
  asset at time $T$.

  We denote its associated state process with initial data $(t,x,y,z)$ by
  $(X^1(s), Y^1(s), Z^1(s))$. Then, because of the existence and
  uniqueness of the strong solution of \eqref{eq:stkprice},
  $$(X^1, Y^1, Z^1) \equiv (X, Y, Z), \ \mathbb{P}-a.s. \hbox{ in }
  A.$$
  Therefore, it is enough to show that
  $$\inf\{s>t: Y^1(s)< c(-Z^1(s))\} = t, \ \mathbb{P}-a.s. $$
  By It\^o's formula, for all $s<\tau_1$ of \eqref{eq:stopping}
  $$d \rho(Y^1(s), Z^1(s)) = d\rho(Y^1(s), z) = - z b(s, X^1(s)) ds +
  z \sigma (s, X^1(s)) d W(s).$$
  By \propref{p-2}, $\inf\{s>t: \rho(Y^1(s), Z^1(s))
  >0\} = t$ under \asmref{a-con}.
\end{proof}

\begin{thm}
  [Continuity of $V$] \label{t-vcon} Assume \asmref{a-con}.
  Then the value function $V$ given in
  \eqref{eq:obj} is  continuous in $(t,x,y)$.
\end{thm}

\begin{proof}
  Fix the initial data $(t,x,y,z) \in [0,T) \times
  \partial^1
  O\cap \{z\neq 0\}$ and arbitrary      $Z\in \mathcal{Z}(t,z)$. Let $Y
  \triangleq Y^{t,x,y,z,Z}$ be a process of
  \eqref{eq:payoff2}.  By   \lemref{l-as}, for any $s\in [t,T)$
  \begin{equation*}
    \int_t^s  \Big(c(-Z(r)) - Y(r)\Big)^+ dr >0 \quad \mathbb{P}-a.s.
  \end{equation*}
  Hence, by definition \eqref{eq:penalty},
  \begin{equation*}\label{eq:t-vcon1}
    \lim_{\varepsilon\to 0^+} \Lambda^\varepsilon(t, s, Y, Z) = 0 \quad
    \mathbb{P}-a.s.
  \end{equation*}

  Fix a small $\delta>0$. Let $Z^\varepsilon
  \in  \mathcal{Z}(t,z)$ be a
  $\delta$-optimal control.
  That is,
  \bea
  V^\varepsilon(t,x,y,z)\ad \le \mathbb{E}_t
  [U(\Lambda^\varepsilon (t, T, Y^{t,x,y,z,Z^\varepsilon},Z^\varepsilon)
  Y^{t,x,y,z,Z^\varepsilon}(T))] + \delta\\
  \ad \triangleq \mathbb{E}_t
  [U(\Lambda^\varepsilon (t, T, Y^{\varepsilon},Z^\varepsilon)
  Y^{\varepsilon}(T))] + \delta
  , \eea
  with the notation $Y^\varepsilon \triangleq
  Y^{t,x,y,z,Z^\varepsilon}$.
  Such a  $\delta$-optimal control
  $Z^\varepsilon$ always exists for each $\varepsilon$.
  Since $V^\varepsilon$ is monotone in $\varepsilon$ and nonnegative,
  $\lim_{\varepsilon \to 0^+} V^\varepsilon(t,x,y,z) $ is
  well-defined. In
  addition, utilizing the fact $\lambda U( y) \le U(\lambda y)$ for
  any $\lambda\in (0,1)$
  \begin{equation*}
    \begin{array}{ll}
      \lim_{\varepsilon \to 0^+} V^\varepsilon(t,x,y,z) &\le
      \lim_{\varepsilon \to 0^+}
      \mathbb{E}_t [ \Lambda^\varepsilon (t, T,
      Y^{\varepsilon},Z^\varepsilon) U( Y^{\varepsilon}(T))] + \delta \\
      & \le       \lim_{\varepsilon \to 0^+}
      \mathbb{E}_t [U(\Lambda^\varepsilon (t, T,
      Y^{\varepsilon},Z^\varepsilon) Y^{\varepsilon, +}(T))] +
      \delta \\
      & \le \lim_\varepsilon U ( \mathbb{E}_t [\Lambda^\varepsilon
      (t, T,
      Y^{\varepsilon},Z^\varepsilon) Y^{\varepsilon, +}(T)]) +
      \delta \\
      &= U ( \lim_\varepsilon  \mathbb{E}_t [\Lambda^\varepsilon (t,
      T, Y^{\varepsilon},Z^\varepsilon) Y^{\varepsilon,+}(T)]) + \delta \\
      & = U ( \mathbb{E}_t [ \lim_\varepsilon  \Lambda^\varepsilon
      (t, T, Y^{\varepsilon},Z^\varepsilon) Y^{\varepsilon,+}(T)]) + \delta   \\
      & = \delta.
    \end{array}
  \end{equation*}
  Note that $V^\varepsilon \ge 0$ and $\delta>0$ is arbitrary.
  These imply the pointwise convergence of
  \begin{equation}
    \label{eq:t-vcon2-1}
    \lim_{\varepsilon \to 0} V^\varepsilon(t,x,y,z) = 0 =
    V(t,x,y,z), \quad \forall (t,x,y,z) \in [0,T) \times \partial^1 O
    \cap \{z \neq 0\}.
  \end{equation}
  It is immediate to show by definition that
  \begin{equation*}
    V^\varepsilon(t,0,c(-z),z) = 0 = V(t,0,c(-z),z), \hbox{ and }
    V^{\varepsilon}(T, x, c(-z),z) = 0 =     V(T, x, c(-z),z) .
  \end{equation*}
  In addition, we can show $\lim_{\varepsilon \to 0} V^\varepsilon
  (t,x,0,0) = 0$ since if $\tau_1$ exists (otherwise trivial)
  \begin{equation*}
    \begin{array}{ll}
      0 &\le \lim_{\varepsilon} V^\varepsilon(t,x,0,0)
      \\ &\le \lim_{\varepsilon} \mathbb{E} [
      V^\varepsilon(\tau_1,X^{t,x}(\tau_1), - c(Z(\tau_1)),
      Z(\tau_1))]
      \\ &\le \mathbb{E} [ \lim_{\varepsilon}
      V^\varepsilon(\tau_1,X^{t,x}(\tau_1), - c(Z(\tau_1)),
      Z(\tau_1))]
      \\ &\le \mathbb{E} [ \lim_{\varepsilon}
      V^\varepsilon(\tau_1,X^{t,x}(\tau_1), c(-Z(\tau_1)),
      Z(\tau_1))]
      \\ & = 0.
    \end{array}
  \end{equation*}
  In the above, we used the
  dominated convergence theorem,
  and applied \eqref{eq:t-vcon2-1} together with the
  fact $Z(\tau_1) \neq 0$. Now,
  we can rewrite \eqref{eq:t-vcon2-1} as
  \begin{equation}
    \label{eq:t-vcon2}
    \lim_{\varepsilon \to 0} V^\varepsilon(t,x,y,z) = 0 =
    V(t,x,y,z), \quad \forall (t,x,y,z) \in [0,T] \times
    \partial^2
    O,
  \end{equation}
  where $\partial^2 O$ is the closure of $\partial^1 O$,
  i.e.,
  \begin{equation*}
    \label{eq:p2Omega}
    \partial^2 O = \{(x,y,z): x\ge 0,
    y = c(-z), z \in \mathcal{K}\}.
  \end{equation*}

  Since $V^\varepsilon(t,x,y,z)$ is continuous on the  compact
  set $([0,T]\times \partial^2 O ) \cap \{x \le \bar x\}$ for
  arbitrary given positive $\bar x$ and converges
  monotonically
  to the zero function by \eqref{eq:t-vcon2},  Dini's theorem  implies
  that
  \begin{equation*}
    \label{eq:t-vcon7}
    \lim_{\varepsilon \to 0^+} V^\varepsilon(t,x,y,z) = 0 \hbox{  uniformly
      on } ([0,T]\times \partial^2 O ) \cap \{x \le \bar x \}.
  \end{equation*}
  Due to the uniform convergence, we can set a real function
  $h^\varepsilon(\cdot) : \mathbb{R}^+ \to  \mathbb{R}^+$ as
  \begin{equation*}
    h^\varepsilon(\bar x) \triangleq \sup\{V^\varepsilon(t,x,y,z):
    (t,x,y,z) \in [0,T]\times \partial^2 O \cap\{x \le \bar
    x\}\}.
  \end{equation*}
  Then
  \begin{equation}
    \label{eq:t-vcon10}
    \lim_{\varepsilon \to 0} h^\varepsilon(\bar x) = 0 \hbox{ for any
      given } \bar x>0.
  \end{equation}
  From \eqref{eq:l-pve1} of \lemref{l-pve}
  and \lemref{l-fox}, we have
  \begin{equation*}
    \lim_{x\to \infty} \frac{h^\varepsilon(x)}{x} = 0 \ \hbox{
      uniformly in } \varepsilon,
  \end{equation*}
  and therefore  there exists a large $x_0>0$ such that
  \begin{equation*}
    \frac{h^\varepsilon(x)}{x} \le \frac{h^\varepsilon(x_0)}{x_0} \
    \hbox{ for all } x\ge x_0
    \ \hbox{ and } \ \varepsilon >0.
  \end{equation*}
  Therefore, we have for all $(t,x,y,z)  \in [0,T] \times \partial^2
  O$
  $$V^\varepsilon(t,x,y,z) \le x \frac{h^\varepsilon(x)}{x} \le x
  \frac{h^\varepsilon(x_0)}{x_0}  = C x h^\varepsilon(x_0).$$
  Now we are ready to
  derive a bound of $V$ in terms of
  $V^\varepsilon$ in the domain
  $(t,x,y,z)\in [0,T]\times O$. First,
  we   observe that, since $\Lambda^\varepsilon(t,s,Y^{t,x,y,z,Z},Z)
  \equiv 1$ for any stopping time $s\le \tau$ of \eqref{eq:tau},
  we can write
  \begin{equation*}
    \label{eq:t-vcon6}
    V^\varepsilon(t,x,y,z)  = \sup_{Z \in \mathcal{Z}(t,z)}
    \mathbb{E}[ V^\varepsilon(\tau,          X^{t,x}(\tau),
    Y^{t,x,y,z,Z}(\tau), Z(\tau)) ].
  \end{equation*}
  Also, the state $(X^{t,x}(\tau), Y^{t,x,y,z,Z}(\tau), Z(\tau)) $
  must fall in $\partial^2 O$,  since $X^{t,x}(\tau) \neq 0$
  almost surely.
  Therefore, for any  $(t,x,y,z)\in  [0,T]\times O$,
  \begin{equation*}\label{eq:t-vcon5}
    \begin{array}{ll}
      V(t,x,y,z)  \le V^\varepsilon(t,x,y,z) \\
      \quad = \sup_{Z \in \mathcal{Z}(t,z)} \mathbb{E}[ V^\varepsilon(\tau,
      X^{t,x}(\tau), Y^{t,x,y,z,Z}(\tau), Z(\tau)) ]
      \\
      \quad \le \sup_{Z \in \mathcal{Z}(t,z)} \mathbb{E}[
      h^\varepsilon(X^{t,x}(\tau)) \one_{\{\tau <T\}} + V^\varepsilon(T,
      X^{t,x}(T), Y^{t,x,y,z,Z}(T), Z(T))  \one_{\{\tau  = T\}} ] \\
      \quad \le C h^\varepsilon(x_0) \mathbb{E}[X^{t,x}(\tau) ] +  \sup_{Z
        \in \mathcal{Z}(t,z)}      \mathbb{E}[ V^\varepsilon(T,
      X^{t,x}(T), Y^{t,x,y,z,Z}(T),
      Z(T))  \one_{\{\tau  = T\}} ] \\
      \quad \le C x h^\varepsilon(x_0) + V(t,x,y,z),
    \end{array}
  \end{equation*} where   $\tau$ is as in \eqref{eq:tau}.
  The above inequalities imply that $V^\varepsilon$ is a
  locally uniform
  estimate of $V$ on the $[0,T]\times O$ in the sense of
  \begin{equation*}
    |V^\varepsilon(t,x,y,z) - V(t,x,y,z)| \le C x  h^\varepsilon(x_0),
    \     \forall (t,x,y,z) \in  [0,T] \times O.
  \end{equation*}
  Finally, we  can  show continuity of $V$ in $(t,x,y)$. For
  any $(t_i,x_i,y_i,z) \in (0,T) \times O$ with $i=1,2$,
  \begin{equation*}
    \begin{array}{ll}
      |V(t_1,x_1,y_1,z) - V(t_2,x_2,y_2,z)| \\
      \le |(V- V^\varepsilon)(t_1,x_1,y_1,z)| +  |(V-
      V^\varepsilon)(t_2,x_2,y_2,z)| +
      |V^\varepsilon(t_1,x_1,y_1,z) -
      V^\varepsilon(t_2,x_2,y_2,z)|  \\
      \le C h^\varepsilon(x_0) (x_1 + x_2 ) +
      |V^\varepsilon(t_1,x_1,y_1,z) -
      V^\varepsilon(t_2,x_2,y_2,z)|.
    \end{array}
  \end{equation*}
  Letting $(t_1,x_1,y_1) \to (t_2,x_2,y_2)$, the last
  term disappears by \lemref{l-pve}, and
  \begin{equation*}
    \lim_{(t_1,x_1,y_1) \to (t_2,x_2,y_2)} |V(t_1,x_1,y_1,z) -
    V(t_2,x_2,y_2,z)|
    \le C x_2 h^\varepsilon(x_0) .
  \end{equation*}
  Thanks to \eqref{eq:t-vcon10}, $\lim_{\varepsilon \to 0}
  h^\varepsilon(x_0) = 0$,  and hence
  \begin{equation*}
    \lim_{(t_1,x_1,y_1) \to (t_2,x_2,y_2)} |V(t_1,x_1,y_1,z) -
    V(t_2,x_2,y_2,z)|  = 0.
  \end{equation*}
\end{proof}

\subsection{Auxiliary Results
  Derived from Continuity} Thanks to the continuity of $V$, now we
can show that the
no-action region is an open set,
which is crucial for the uniqueness. (see inequalities
\eqref{eq:t-cpr2} and \eqref{eq:t-cpr3} with application of Ishii's lemma)
\begin{prop}
  \label{p-nonaction}
  $\mathcal{A}[V](z)$ is open in $\mathbb{R}^3$ for any $z\in
  \mathcal{K}$.
\end{prop}
\begin{proof}
  By the definition of $\mathcal{A}$ of \eqref{eq:nonaction}, we write
  \begin{equation*}
    \mathcal{A}[V] (z) = \{(t,x,y): V(t,x,y,z) >
    \mathcal{S} V(t,x,y,z) \} = \{(t,x,y): (V -
    \mathcal{S} V)(t,x,y,z) >0  \}.
  \end{equation*}
  Note that
  $V(\cdot, \cdot, \cdot, z)$ is continuous by \thmref{t-vcon}, so is
  $(V- \mathcal{S}V) (\cdot, \cdot, \cdot, z)$ by
  \propref{p-opS}. This implies $\mathcal{A}[V] (z)$ is an open set.
\end{proof}

\propref{p-nonaction} also enables us to characterize the  optimal
strategy by a $\mathbb{F}^X$-predictable process, where $\mathbb{F}^X$
is filtration generated by price process $X$. Practically, a trader
can observe only the price process $X$ (not the Brownian motion $W$), and
$\mathbb{F}^X$-predictable strategy is more desirable. We briefly
discuss the construction of the optimal strategy below.

By standard argument, the optimal strategy is essentially constructed
by a series of optimal stopping time problem. Indeed, given initial
state $(t,x,y,z)$, using $Y_1$ to denote the process $Y^{t,x,y,Z}$
with constant control $Z \equiv z$, the first transaction is occurred
at
\begin{equation}
  \label{eq:tau1-1}
  \tau_1 = \inf\{s\ge t: Y_1 \notin \mathcal{A}[V](z)\}
\end{equation}
and the size of transaction at $\tau_1$ is
\begin{equation}
  \label{eq:size1}
  Z(\tau_1) - Z(\tau_1^-) = \arg\max_{\Delta z} V(\tau_1^-, X(\tau_1^-),
  Y_1(\tau_1^-)- c(\Delta z), z+ \Delta z).
\end{equation}
The subsequent transaction times and sizes are determined repeatedly by
using the same procedure.

Note that, since $\mathcal{A}[V](z)$ is an open set by
\propref{p-nonaction},
$\tau_1$ of \eqref{eq:tau1-1}  is an $\mathbb{F}^{Y_1}$-stopping
time, where $\mathbb{F}^{Y_1}$ is the filtration generated by
$Y_1$. Furthermore, together with the fact $\mathbb{F}^{Y_1} \subset
\mathbb{F}^X$, this implies that $\tau_1$ is
an $\mathbb{F}^X$-stopping time. Also
note that in \eqref{eq:size1}, the jump size of $Z(\tau_1) -
Z(\tau_1^-)$ is measurable with respect to
$\mathbb{F}^X(\tau_1^-)$. Repeating above argument to the subsequent
jump times, one can show that
the above constructed process is $\mathbb{F}^X$-predictable.

\section{Characterization of  Value Function}\label{sec:val}
In this section, we will show the value function is the unique
viscosity solution of \eqref{eq:qvi} with condition \eqref{eq:btc}.
First, we give definition of viscosity solution:
\begin{defn}
  \label{d-visc}  {\rm A function
    $u$ is said to be
    a viscosity subsolution (resp. supersolution) of
    \eqref{eq:qvi}-\eqref{eq:btc}, if
    \begin{enumerate}
    \item for any $(t_0,x_0,y_0,z_0) \in (0,T) \times O$ and function
      $\varphi \in C^{1,2,2}((0,T)\times O, \mathbb{R})$ satisfying
      \begin{equation*}
        \label{eq:d-visc1}
        \varphi \ge (\hbox{resp. } \le) \  u \hbox{ on } (0,T) \times
        O \ \hbox{ and }   \varphi = u \hbox{ at } (t_0,x_0,y_0,z_0),
      \end{equation*}
      the  following inequality holds:
      \begin{equation*}
        \label{eq:d-visc2}
        \min\{(-\varphi_t - \mathcal{L} \varphi)
        (t_0, x_0, y_0, z_0), (\varphi
        - \mathcal{S} \varphi) (t_0, x_0, y_0, z_0)\} \le (\hbox{resp. }
        \ge)  0,
      \end{equation*}
      and
    \item
      $u(t,x,y,z) \le (\hbox{resp. } \ge)\  U(y - c(-z)), \hbox{ on
      } \partial^* ([0,T) \times O).$
    \end{enumerate}
    The $u$ is said to be a viscosity solution, if it is both a
    viscosity subsolution and a viscosity supersolution. }\end{defn}

\subsection{Viscosity Solution Properties}

Next, we show the objective function $V$ of
\eqref{eq:obj} is a viscosity solution of quasi-variational inequality
\eqref{eq:qvi}-\eqref{eq:btc}.
\begin{thm}[Viscosity properties]
  \label{t-existence}
  The objective function $V(t,x,y,z)$ of \eqref{eq:obj} is a viscosity
  solution of the quasi-variational inequality
  \eqref{eq:qvi} with boundary-terminal condition \eqref{eq:btc}.
\end{thm}

\begin{proof} The proof is divided into two steps.
  \begin{enumerate}
  \item First, we prove
    that $V$ is a supersolution of
    \eqref{eq:qvi}.
    If not, there would exist $(t_0,
    \eta_0) \triangleq (t_0, x_0, y_0, z_0)$ and a
    function
    $\varphi \in C^{1,2,2}((0,T)\times O, \mathbb{R})$ with
    $$\varphi \le V, \hbox{ and } \varphi (t_0, \eta_0) = V (t_0,
    \eta_0)$$
    satisfying
    \begin{equation}
      \label{eq:t-existence1}
      \min\{(-\varphi_t - \mathcal{L} \varphi)
      (t_0, \eta_0), (\varphi
      - \mathcal{S} \varphi) (t_0, \eta_0)\} < 0.
    \end{equation}
    Since by \eqref{eq:dpp1} and monotonicity of
    \propref{p-opS},
    $$\varphi(t_0, \eta_0) = V(t_0, \eta_0) \ge \mathcal{S} V(t_0,
    \eta_0) \ge \mathcal{S} \varphi(t_0, \eta_0),$$  and
    \eqref{eq:t-existence1} is equivalent to
    \begin{equation}
      \label{eq:t-existence2}
      (-\varphi_t - \mathcal{L} \varphi)(t_0, \eta_0) <0.
    \end{equation}
    We introduce a
    strict subtest
    function $\phi(\cdot)$ given by
    \begin{equation*}
      \phi(t,x,y,z) = \varphi(t,x,y,z) - |t - t_0|^2 - |x - x_0|^4 -
      |y - y_0|^4.
    \end{equation*}
    One can check $\phi$ also satisfies  inequality
    \eqref{eq:t-existence2}, i.e.,
    \begin{equation*}
      \label{eq:t-existence3}
      (-\phi_t - \mathcal{L} \phi)(t_0, \eta_0) <0.
    \end{equation*}
    Since $-\phi_t - \mathcal{L} \phi$ is continuous in $(t,x,y)$,
    \begin{equation*}
      \{(t,x,y): (-\phi_t - \mathcal{L}\phi)(t,x,y,z_0)<0\}
    \end{equation*}
    is an open set. Now, we can take a small open ball $B_r(t_0, x_0,
    y_0) \times \{z_0\} \subset (0,T) \times O$ such that
    \begin{equation*}
      (-\phi_t -\mathcal{L} \phi)(t,\eta) <0, \hbox{ in } B_r(t_0,
      x_0, y_0) \times \{z_0\}.
    \end{equation*}
    Observe that $\forall (t,\eta) \in \partial B_r(t_0, x_0, y_0)
    \times \{z_0\}$
    \begin{equation}
      \label{eq:t-existence4}
      \varphi(t,\eta) - \phi(t,\eta) = |t-t_0|^2 + |x- x_0|^4 +
      |y-y_0|^4 \ge 1 \wedge \frac{r^4}{3} \triangleq \varepsilon.
    \end{equation}
    Consider the stopping time $\theta$ defined by
    \begin{equation*}
      \theta = \{ s\ge t_0: (s, X^{t_0,x_0}(s), Y^{t_0,x_0,y_0,z_0}(s)
      \notin B_r(t_0,x_0,y_0) \}.
    \end{equation*}
    Applying It\^{o}'s formula on $\phi$, with notations $X^{t_0,x_0}
    \triangleq X$, $Y^{t_0,x_0,y_0,Z_0} = Y$, and $Z_0(\cdot) \equiv
    z_0$, we have
    \begin{equation*}
      \label{eq:t-existence5}
      \begin{array}{ll}
        V(t_0,\eta_0) & = \phi(t_0,\eta_0) \\
        & = \mathbb{E}[ \phi(\theta, X(\theta), Y(\theta), z_0) -
        \disp\int_{t_0}^{\theta} (\phi_t + \mathcal{L} \phi) (s, X(s),
        Y(s), z_0) ds] \\&
        \le \mathbb{E}[\phi(\theta, X(\theta), Y(\theta), z_0)] \\&
        \le \mathbb{E}[\varphi(\theta, X(\theta), Y(\theta), z_0)] -
        \varepsilon \\&
        \le \mathbb{E}[V(\theta, X(\theta), Y(\theta), z_0)] -
        \varepsilon \\ & \le V(t_0, \eta_0) - \varepsilon.
      \end{array}
    \end{equation*}
    This leads to  a contradiction and completes the proof of
    viscosity supersolution property.
  \item Next, we show the viscosity subsolution property. To the
    contrary, if there exists $(t_0, \eta_0) \triangleq (t_0, x_0,
    y_0, z_0)$ and a function $\varphi \in C^{1,2,2}((0,T)
    \times O,  \mathbb{R})$ with
    \begin{equation*}
      \varphi \ge V, \hbox{ and } \varphi(t_0, \eta_0) = V(t_0,
      \eta_0)
    \end{equation*}
    satisfying
    \begin{equation*}
      \min\{(-\varphi_t - \mathcal{L} \varphi) (t_0, \eta_0), (\varphi
      - \mathcal{S} \varphi) (t_0, \eta_0)\} > 0.
    \end{equation*}
    One can rewrite the above inequality as
    \begin{equation}\label{eq:t-existence6}
      (-\varphi_t - \mathcal{L} \varphi) (t_0, \eta_0)>0, \quad  (\varphi
      - \mathcal{S} \varphi) (t_0, \eta_0) > 0.
    \end{equation}
    The second inequality of \eqref{eq:t-existence6}, together with the
    monotonicity of $\mathcal{S}$ of \propref{p-opS}, leads to
    \begin{equation*}
      V(t_0, \eta_0) = \varphi(t_0, \eta_0) > \mathcal{S} \varphi(t_0,
      \eta_0) \ge \mathcal{S} V(t_0, \eta_0),
    \end{equation*}
    that is equivalent to
    \begin{equation}
      \label{eq:t-existence7}
      (t_0, x_0, y_0) \in \mathcal{A}[V](z_0),
    \end{equation}
    Now, we consider a
    test function $\phi$ given by
    \begin{equation*}
      \phi(t,x,y,z) = \varphi(t,x,y,z) + |t- t_0|^2
      + |x-x_0|^4 + |y -
      y_0|^4.
    \end{equation*}
    One can check that, by \eqref{eq:t-existence6}
    \begin{equation*}
      (-\phi_t - \mathcal{L} \phi)(t_0, \eta_0) >0.
    \end{equation*}
    Since $(-\phi_t - \mathcal{L} \phi)$ is continuous in $(t,x,y)$,
    \begin{equation*}
      \{(t,x,y): (-\phi_t -\mathcal{L}\phi)(t,x,y,z_0)>0\}
    \end{equation*}
    is an open set. Note also that
    \eqref{eq:t-existence7} together with
    \propref{p-nonaction} implies $\mathcal{A}[V](z_0)$ is
    a non-empty
    open set. Thus,
    \begin{equation}\label{eq:t-existence8}
      \{(t,x,y): (-\phi_t -\mathcal{L}\phi)(t,x,y,z_0)>0\} \cap
      \mathcal{A}[V](z_0)
    \end{equation}
    is also a non-empty set. We can
    take a small open
    ball
    $B_r(t_0, x_0, y_0) \times \{z_0\}$ contained in the open set of
    \eqref{eq:t-existence8}, i.e.,
    \begin{equation*}\label{eq:t-existence9}
      (-\phi_t - \mathcal{L}\phi)(t,\eta)>0, \quad V(t,\eta) >
      \mathcal{S} V(t,\eta), \quad \forall (t,\eta) \in B_r (t_0, x_0,
      y_0) \times \{z_0\}.
    \end{equation*}
    Similar to \eqref{eq:t-existence4}, we also have
    \begin{equation*}
      \phi(t,\eta) - \varphi(t,\eta) = |t- t_0|^2 + |x-x_0|^4 +
      |y-y_0|^4 \ge 1\wedge \frac{r^4}{3} \triangleq \varepsilon.
    \end{equation*}
    Define
    \begin{equation*}
      \theta = \{ s\ge t_0: (s, X^{t_0,x_0}(s), Y^{t_0,x_0,y_0,z_0}(s)
      \notin B_r(t_0,x_0,y_0) \}.
    \end{equation*}
    Applying It\^o's formula to
    $\phi$, with notations $X^{t_0,x_0}
    \triangleq X$, $Y^{t_0,x_0,y_0,Z_0} = Y$, and $Z_0(\cdot) \equiv
    z_0$, we obtain
    \begin{equation*}
      \begin{array}{ll}
        V(t_0,\eta_0) & = \phi(t_0,\eta_0) \\
        & = \mathbb{E}[ \phi(\theta, X(\theta), Y(\theta), z_0) -
        \disp \int_{t_0}^{\theta} (\phi_t + \mathcal{L} \phi) (s, X(s),
        Y(s), z_0) ds] \\&
        \ge \mathbb{E}[\phi(\theta, X(\theta), Y(\theta), z_0)] \\&
        \ge \mathbb{E}[\varphi(\theta, X(\theta), Y(\theta), z_0)] +
        \varepsilon \\&
        \ge \mathbb{E}[V(\theta, X(\theta), Y(\theta), z_0)] +
        \varepsilon.
      \end{array}
    \end{equation*}
    Since
    $V(t_0, \eta_0) = \mathbb{E}[V(\theta, X(\theta), Y(\theta), z_0)]
    $ in the no-action region $\mathcal{A}[V](z_0)$, this leads to a
    contradiction.
  \end{enumerate}
\end{proof}

\subsection{Uniqueness}
In this part, we establish
the uniqueness
in the sense of viscosity solution for the
quasi-variational inequality \eqref{eq:qvi} with boundary-terminal
condition \eqref{eq:btc}.

Throughout this section, we assume that $u$ and $v$ are continuous
sub- and supersolution of \eqref{eq:qvi} and \eqref{eq:btc},
respectively, satisfying sublinear growth of the form, for
$\varphi = u, v$
\begin{equation}
  \label{eq:sublinear}
  \lim_{x \to \infty} \sup_{t} \frac{\varphi(t,x,y,z)}{x}
  = 0, \forall (y,z), \hbox{ and }   \lim_{y \to \infty}
  \sup_{t} \frac{\varphi(t,x,y,z)}{y} = 0, \forall (x,z).
\end{equation}
We are to show a comparison result
$$u\ge v $$
on the entire domain, which implies uniqueness.

\begin{asm}
  \label{a-bsbd} {\rm
    The $b$ and $\sigma$ are uniformly bounded,
    i.e., there exists a positive
    constant $C_4$ such that $ \sup_{[0,T]\times [0,\infty)} |b(t,x)| +
    |\sigma(t,x)| < C_4.$ }
\end{asm}

Define constants
$$\rho = \frac 1 2 \min_{z\neq 0} c(z)>0,
\ \hbox{ and }\ C_5 = \|b\|_\infty
(C_2 \vee C_3 +1) + 2 \rho $$
and
\begin{equation*}
  \label{eq:ve}
  v^\varepsilon(t,x,y,z) = v(t,x,y,z) + \varepsilon g(t,x,y,z)
\end{equation*}
where $g(t,x,y,z) = x + y + C_5(T - t)$.

\begin{lem}
  \label{l-ssup}
  $v^\varepsilon$ is a strict supersolution, i.e., any
  smooth
  test
  function $\varphi^\varepsilon$ with $\varphi^\varepsilon =
  v^\varepsilon$ at $(\bar t, \bar x, \bar y, \bar z)\in (0,T)\times
  O$ satisfies
  \begin{equation}
    \label{eq:l-ssup1}
    (\varphi^\varepsilon - \mathcal{S} \varphi^\varepsilon)(\bar t, \bar
    x, \bar y, \bar z) >
    \varepsilon \rho > 0,
  \end{equation}
  and
  \begin{equation}
    \label{eq:l-ssup2}
    (-\varphi^\varepsilon_t - \mathcal{L} \varphi^\varepsilon) (\bar t,
    \bar x, \bar y, \bar z)
    > \varepsilon \rho>0.
  \end{equation}
\end{lem}

\begin{proof} Note that $\varphi \triangleq \varphi^\varepsilon -
  \varepsilon g$
  is a
  test function of $v$ at $(\bar t, \bar x, \bar y, \bar z)$,
  and by viscosity supersolution property
  $$\min\{ (-\varphi_t - \mathcal{L} \varphi)(\bar t, \bar x, \bar y,
  \bar  z),  (\varphi - \mathcal{S} \varphi) (\bar t, \bar x, \bar y,
  \bar
  z) \} \ge 0,$$
  Using \propref{p-opS}, \eqref{eq:l-ssup1} is obtained from
  \begin{equation*}
    (\varphi^\varepsilon - \mathcal{S} \varphi^\varepsilon) (\bar t,
    \bar x, \bar y, \bar z) \ge (\varphi - \mathcal{S} \varphi) (\bar
    t, \bar x, \bar y, \bar z) + \varepsilon (g - \mathcal{S} g) (\bar
    t, \bar x, \bar y, \bar z)  \ge \varepsilon \rho.
  \end{equation*}
  Equation
  \eqref{eq:l-ssup2} is the result of viscosity supersolution property
  of $v$ and
  \begin{equation*}
    g_t + \mathcal{L}g < - \rho.
  \end{equation*}
\end{proof}

\begin{lem}
  \label{l-funH}
  Let
  \begin{equation}
    \label{eq:l-funH1}
    H(t,x,y,z) = u(t,x,y,z) - v^\varepsilon (t,x,y,z).
  \end{equation}
  Then $H(t,x,y,z)$ attains its maximum in
  $[0,T] \times \bar O$,
  i.e.,
  $\exists (\hat t, \hat x, \hat y, \hat z) \in [0,T] \times \bar
  O$ such that
  \begin{equation*}
    H(\hat t, \hat x, \hat y, \hat z) = \max_{[0,T]\times \bar O} H(t,x,y,z).
  \end{equation*}
  Moreover,
  \begin{equation*}
    \label{eq:l-funH2}
    (u - \mathcal{S} u) (\hat t,\hat x, \hat y, \hat z) > 0.
  \end{equation*}
\end{lem}

\begin{proof} Since $\varepsilon g$
  grows at most linearly and
  $u-v$ has sublinear growth of the form \eqref{eq:sublinear} in
  $(x,y)$, $H$ satisfies $H(t,x,y,z) \to -\infty$ as $|x| +  |y| \to
  \infty$. Thus, $H(\cdot)$ attains its maximum at some point in
  its domain, say $(\hat t, \hat x, \hat y, \hat z)$. To the contrary, if
  $$(u-v)(\hat t,\hat x, \hat y, \hat z)  \le 0.$$
  Then,  $\exists z^*
  \neq \hat z$  such that
  $$u(\hat t,\hat x, \hat y, \hat z)   \le u (\hat t,\hat x, \hat y -
  c(z^* - \hat z),   z^*). $$
  On the other hand, by \propref{p-opS}- monotonicity and
  \lemref{l-ssup},
  $$v^\varepsilon (\hat t, \hat x, \hat y, \hat z) > \mathcal{S}
  v^\varepsilon (\hat t, \hat x, \hat y, \hat z) \ge v^\varepsilon
  (\hat t,\hat x, \hat y - c(z^* - \hat z),   z^*) .$$
  Combining the above two inequalities,
  $$(u-v^\varepsilon) (\hat t, \hat x, \hat y, \hat z) <
  (u-v^\varepsilon) (\hat t,\hat x, \hat y - c(z^* - \hat z),
  z^*),$$ which is a contradiction to
  the definition of $(\hat t, \hat x,
  \hat y, \hat z)$ as a maximizer.
\end{proof}

For any $z\in \mathcal{K}$, define
\begin{equation}
  \label{eq:Omegaz}
  O_z = \{(x,y): (x,y,z) \in O\}.
\end{equation}

\begin{lem}
  \label{l-phia}
  Define $\Phi_\alpha: [0,T] \times
  O_z^2 \times \mathcal{K} \to
  \mathbb{R}$ as
  \begin{equation*}
    \label{eq:l-phia1}
    \Phi_\alpha (t, \zeta, \eta, z) = u(t,\zeta,z) - v^\varepsilon
    (t,\eta,z) - \frac \alpha 2 |\zeta - \eta|^2.
  \end{equation*}
  Then, the following assertions are true:
  \begin{enumerate}
  \item For
    each  $z\in \mathcal{K}$,
    $\Phi_\alpha (\cdot,z)$ achieves its maximum at a point in
    $[0,T]\times \bar O_z^2$,  denoted
    by $(t_\alpha^z, \zeta_\alpha^z, \eta_\alpha^z)$.
  \item
    There exists a convergent
    subsequence $(t_\alpha^z, \zeta_\alpha^z)
    \to (t^z, \zeta^z) \in [0,T] \times O_z$  such that
    \begin{equation*}
      \label{eq:l-phia2}
      H(t^z, \zeta^z,z) = \sup_{[0,T] \times \bar O_z} H(t,\zeta,z),
    \end{equation*}
    and
    \begin{equation*}
      \label{eq:l-phia3}
      \lim_{\alpha\to \infty} \alpha |\zeta_\alpha^z -
      \eta_\alpha^z|^2 \to 0.
    \end{equation*}
  \end{enumerate}
\end{lem}

\begin{proof} Note that
  $\Phi_\alpha (\cdot,z) $ achieves maximum, since
  $\Phi_\alpha(t,\zeta,\eta,z) \to -\infty$ as $|\zeta| + |\eta| \to
  \infty$. The rest of proof is an application of \cite[Lemma
  3.1]{CIL92} on function $\Phi_\alpha(\cdot, \cdot, \cdot, z)$.
\end{proof}

\begin{thm}[Comparison result]\label{t-cpr}
  $$\sup_{[0,T]\times \bar O} (u-v) (t, x,y,z) = \sup_{\partial^*
    ([0,T)\times  O)} (u-v)(t, x,y,z).$$
\end{thm}

\begin{proof}
  It suffices
  to show that for an arbitrary given $\varepsilon$,
  \begin{equation*}
    \label{eq:t-cpr1}
    \sup_{[0,T]\times \bar O} H(t, x,y,z) =
    \sup_{\partial^* ([0,T)\times O)} H(t,x,y,z).
  \end{equation*}
  To the contrary, we assume $$H(\hat t, \hat x, \hat y, \hat z) =
  \sup_{[0,T]\times \bar O}
  H(t,x,y,z)>\sup_{\partial^* ([0,T)\times O)} H(t,x,y,z)$$ for some
  $\varepsilon>0$.
  Then, we have
  $(\hat t, \hat x, \hat y, \hat z) \in [0,T)\times O$ at the
  interior   of the domain.

  For notational convenience, we denote $(\hat t_\alpha, \hat
  \zeta_\alpha, \hat \eta_\alpha) = (t_\alpha^{\hat z}, \zeta_\alpha^{\hat z},
  \eta_\alpha^{\hat z})$. Also, we note that $(\hat t, \hat \zeta) =
  (t^{\hat z}, \zeta^{\hat z}) = \lim_{\alpha \to \infty} (\hat
  t_\alpha, \hat \zeta_\alpha)$ in \lemref{l-phia}. By
  \lemref{l-phia}, we have
  $\Phi_\alpha  (\hat t_\alpha, \hat \zeta_\alpha, \hat \eta_\alpha,
  \hat z) \to H
  (\hat t, \hat \zeta, \hat z)$ as $\alpha \to \infty$.

  By \lemref{l-funH}, $(\hat t, \hat \zeta) \in \mathcal{A}[u](\hat
  z).$ Since  $\mathcal{A}[u](\hat z)$ is open by
  \propref{p-nonaction},  there
  exists some
  $\alpha_1>0$ such that $(\hat t_\alpha, \hat \zeta_\alpha), (\hat
  t_\alpha, \hat \eta_\alpha) \in \mathcal{A}[u](\hat z) \subset [0,T)\times
  O_{\hat z} $ for all $\alpha>\alpha_1$.
  To proceed, we denote parabolic superjet (resp. subjet) by
  $D^{+(1,2)}$ (resp. $D^{-(1,2)}$), and its closure by $\bar D^{+(1,2)}$
  (resp. $\bar D^{-(1,2)}$); see its definition and properties in
  \cite{CIL92}.  Applying Ishii's lemma (also in \cite{CIL92}) on
  $u(\cdot,{\hat z}),
  v^{\varepsilon}(\cdot, \hat z)$,  and $w^\alpha(t,\zeta,\eta) =
  \frac \alpha 2 |\zeta - \eta|^2$, there exists $q, \tilde q \in
  \mathbb{R}$, $p, \tilde p \in \mathbb{R}^2$ and  symmetric matrices
  $A, B$ depending on $\alpha$, such that
  \begin{enumerate}
  \item $(q, p, A) \in \bar D^{+(1,2)} u(\hat t_\alpha, \hat
    \zeta_\alpha, {\hat z})$, $p = D_{\zeta} w^\alpha(\hat t_\alpha,
    \hat \zeta_\alpha, \hat \eta_\alpha) = \alpha (\hat \zeta_\alpha -
    \hat \eta_\alpha)$;
  \item $(\tilde q, \tilde p, \tilde A) \in \bar D^{-(1,2)}
    v^{\varepsilon}(\hat t_\alpha, \hat \eta_\alpha, \hat z)$, $\tilde
    p = -D_{\eta} w^\alpha(\hat t_\alpha, \hat \zeta_\alpha, \hat
    \eta_\alpha) = \alpha (\hat \zeta_\alpha - \hat \eta_\alpha)$;
  \item $q - \tilde q = 0$;
  \item $-3\alpha I_4 \le \left[
      \begin{array}{ll}
        A & 0 \\ 0 & -\tilde A
      \end{array} \right] \le
    3\alpha \left[
      \begin{array}{ll} I_2 & -I_2 \\ -I_2 & I_2
      \end{array} \right].$
  \end{enumerate}
  By viscosity subsolution property of $u$, it yields
  $$\min\{ F(\hat t_\alpha, \hat \zeta_\alpha, \hat z, q, p, A), (u -
  \mathcal{S} u)(t_\alpha, \hat \zeta_\alpha, \hat z) \}  \le  0.$$
  Since $(\hat t_\alpha, \hat \zeta_\alpha) \in  \mathcal{A}[u](\hat
  z)$ and $\mathcal{A}[u](\hat z)$ is open
  \begin{equation}
    \label{eq:t-cpr2}
    F(\hat t_\alpha, \hat \zeta_\alpha, \hat z, q, p, A) \le 0.
  \end{equation}
  Also by \eqref{eq:l-ssup2},
  \begin{equation}
    \label{eq:t-cpr3}
    F(\hat t_\alpha, \hat \eta_\alpha, \hat z, \tilde q, \tilde p,
    \tilde A) >\varepsilon \rho>0.
  \end{equation}
  Using the result of \lemref{l-phia}, Lipschitz condition on $b$ and
  $\sigma$, and
  Ishii's lemma, subtracting
  \eqref{eq:t-cpr3} from \eqref{eq:t-cpr2}
  \begin{equation*}
    \varepsilon \rho < F(\hat t_\alpha, \hat \eta_\alpha, \hat z,
    \tilde q, \tilde p, \tilde A) - F(\hat t_\alpha, \hat
    \zeta_\alpha, \hat z, q, p, A) \to 0
    \hbox{ as } \alpha \to \infty,
  \end{equation*}
  which leads to a contradiction.
\end{proof}

\subsection{Summary of Results}
Finally, we summarize what have been obtained so far. It is
presented in the
following characterization of the value function.

\begin{thm}
  \label{t-characterization}
  Given \asmref{a-general}, \asmref{a-con} and \asmref{a-bsbd},  the
  value function
  $V(\cdot, \cdot, \cdot, \cdot)$ of \eqref{eq:obj} is
  the unique
  viscosity solution of the quasi-variational inequality
  \eqref{eq:qvi}-\eqref{eq:btc} in the space of continuous functions
  with sublinear growth in $(x,y)$ of the form \eqref{eq:sublinear}.
\end{thm}

\section{Further Remark}\label{sec:rem}
In this work, we obtained the continuity of the value function,
and further characterized the value function as the unique
viscosity solution of a quasi-variational inequality with
Cauchy-Dirichlet condition on $\partial^* ([0,T) \times O)$ under
some appropriate assumptions.

We have emphasized the continuity result in the current work. As a
future study, we will consider viable uniqueness proofs with
boundary conditions only on the effective boundary $\partial^*
([0,T) \times O) \cap \{x>0\}$. The other consideration is to show
the uniqueness without \asmref{a-bsbd}. One possible approach is
to use domain transformation defined by $\bar x = \ln x$, and
adjust the operators $\mathcal{L}$ and $\mathcal{S}$
appropriately, which is not included in the current paper due to
notational complexity.

Another possible extension of the current work is to consider
transaction cost of the form $c(x,z)$, with subadditive condition in
$z$. More discussions are referred to \cite{MSXZ08a}. It might also
be
interesting to study regime-switching models under optimal
switching framework.

It is straightforward to generalize all the results to nonzero
fixed risk-free rate $r>0$ by usual normalization. However, it is
nontrivial to consider similar utility maximization problems under
various stochastic interest rate models.

\section{Appendix}\label{sec:app}

Next, for the sake of  completeness, we show the sample path results
on 1-D It\^o's process. \propref{p-samplepath} is a generalized version
of \cite{BSY09}, and \propref{p-2} is a special case of
\propref{p-samplepath}, which is needed in the proof of \lemref{l-as}.

Consider the oem-dimensional It\^o process
\begin{equation}
  \label{eq:itoX}
  X(t,\omega) = \int_0^t b(s,\omega) ds + \sigma(s,\omega) d W(s),
\end{equation} where we assume \begin{equation}
  \label{eq:asmbs1}
  b(\cdot, \omega) \in L^1([0,T]), \sigma(\cdot, \omega) \in
  L^2([0,T]), \ \mathbb{P}-a.s.
\end{equation}
The
It\^o process \eqref{eq:itoX}
is well defined
under assumption \eqref{eq:asmbs1}; see Definition 4.1.1 of
\cite{Oks03} and Problem 4.11 of \cite{KS91}.

Define the stopping times
\begin{equation}
  \label{eq:eta}
  \eta(\omega) \triangleq \inf\{t>0: \int_0^t \sigma^2(s, \omega) ds
  >0 \},
\end{equation}
and
\begin{equation*}
  \label{eq:tau1}
  \tau(\omega) \triangleq \inf\{t>0: X(t) >0 \}.
\end{equation*}

\begin{prop}
  \label{p-samplepath}
  Assume \eqref{eq:asmbs1} holds. For
  the It\^o process given by \eqref{eq:itoX},
  $\tau = 0$ $\mathbb{P}$-a.s. if
  one of the following two conditions is satisfied:
  \begin{enumerate}
  \item There exists stopping time  $\theta >0$ such that
    \begin{equation}
      \label{eq:theta}
      \int_0^t b(s, \omega) ds >0,  \ \forall t< \theta(\omega);
    \end{equation}
  \item
    $\eta = 0$, and there exists measurable function $\psi$ satisfying
    \begin{equation}
      \label{eq:psi}
      b = \sigma \psi, \ \hbox{ and }  \mathbb{E} \Big[\exp \Big\{
      \frac 1 2 \int_0^T
      \|\psi\|^2 ds \Big\} \Big] < \infty, \ \hbox{ for some } T>0.
    \end{equation}
  \end{enumerate}
\end{prop}

\begin{proof} We divide the proof into three steps:
  \begin{enumerate}
  \item  Suppose $b \equiv 0$ and $\eta = 0$ almost
    surely. Define $$\eta_\varepsilon \triangleq \inf\{t>0: \int_0^t
    \sigma^2 (s, \omega) ds > \varepsilon \}.$$ Then,
    $\eta_\varepsilon \to 0$ as $\varepsilon \to 0$
    $\mathbb{P}$-a.s. Set $X_\varepsilon(t) = \int_0^t
    \sigma_\varepsilon(s,\omega) dW(s)$, where
    $$\sigma_\varepsilon(s, \omega) = \left\{
      \begin{array}{ll}
        \sigma(s, \omega), & s \le \eta_\varepsilon(\omega) \\
        1, & s> \eta_\varepsilon(\omega).
      \end{array}\right.
    $$ Then,
    the quadratic variation
    $\langle X_\varepsilon \rangle (t) \to \infty$ as $t\to
    \infty$, and $X_\varepsilon (s, \omega) = X(s, \omega)$ for all
    $s < \eta_\varepsilon(\omega)$.

    Let $T_\varepsilon(s) = \inf\{t>0: \langle X_\varepsilon \rangle
    (t)>s\}$.
    By Theorem 4.6 of \cite{KS91}, $B(s) \triangleq
    X_\varepsilon(T_\varepsilon(s))$ is a Brownian motion under
    $\mathbb{P}$ with time-changed filtration.

    Note that $T_\varepsilon (\varepsilon) = \eta_\varepsilon$ by
    definition, and hence $B(\varepsilon) = X_\varepsilon
    (T_\varepsilon(\varepsilon)) = X_\varepsilon (\eta_\varepsilon)$
    and
    \begin{equation*}
      \sup_{0\le t\le \eta_\varepsilon} X(t, \omega) \equiv
      \sup_{0\le t\le \eta_\varepsilon} X_\varepsilon (t, \omega) =
      \sup_{0\le t \le \varepsilon} B(t,\omega) >0, \
      \mathbb{P}-a.s. \ \forall \varepsilon>0.
    \end{equation*}
    Therefore, we obtain
    \bea
    0\ad \le \inf\{t>0: X(t, \omega) >0 \}
    \\
    \ad \le \inf\{s>0:
    \sup_{0\le s\le t}  X(t, \omega) >0 \} \le \eta_\varepsilon \to
    0, \ \varepsilon \to 0, \ \mathbb{P}-a.s.
    \eea
    and this implies $\tau = 0$ if $b \equiv 0$ and $\eta = 0$.
  \item
    Now we assume existence of $\theta > 0$ satisfying
    \eqref{eq:theta}. Set $A = \{ \omega: \eta(\omega) >0 \}$,
    then
    $$X(t,\omega) = \int_0^t b(s, \omega) ds, \ \forall s<
    \eta(\omega), \ \mathbb{P}-a.s. \ \omega \in A,$$
    and thus $\tau = 0$ $\mathbb{P}$-a.s. $\omega \in A$. For  $\omega
    \in A^c$, we have following inequality
    $$X(t, \omega) \ge \int_0^t \sigma(s,\omega) dW(s), \ \forall t<
    \theta(\omega),$$
    and by the proof of
    the first part, since $\eta = 0$ $\mathbb{P}$-a.s. in
    $A^c$,
    $$\sup_{0\le s \le t} X(s, \omega) \ge \sup_{0\le s \le t}
    \int_0^s \sigma(r,\omega) dW(r) >0 , \ \forall t<
    \theta(\omega).$$
    This implies $\tau = 0$, $\mathbb{P}$-a.s. in $A^c$.
    Thus, $\tau = 0$ with existence of $\theta >0$ of
    \eqref{eq:theta}.
  \item
    With the assumption of \eqref{eq:psi}, by Girsanov theorem, $X(t)
    = \int_0^t \sigma (s, \omega) d \widetilde W^Q$, where $Q \sim
    \mathbb{P}$. Then we can  apply the result of first part  to obtain
    $\tau = 0$ $Q$-a.s., and hence $\mathbb{P}$-a.s.
  \end{enumerate}
\end{proof}

\begin{exm}
  \label{e-1}  {\rm
    Suppose $X(t) = -100 t + W(t)$, then $\tau = 0$. This can be seen from
    (2) of \propref{p-samplepath}.}
\end{exm}

\begin{prop}
  \label{p-2}
  Assuming that $\mathcal{F}_0$ is trivial $\sigma$-algebra.
  Consider It\^o process \eqref{eq:itoX} with continuous $b(\cdot, \omega)$
  and $\sigma(\cdot,\omega)$. Then, $\tau = 0$, if either $b(0)>0$ or
  $\sigma(0)>0$.
\end{prop}

\begin{proof}
  The proof is carried out in two steps.

  \begin{enumerate}
  \item   If $b(0)>0$,
    we can take $\theta = \inf\{t>0: b(t,\omega) \le \frac
    1 2 b(0).\}$. Because of the continuity of $b$, $\theta$ satisfies
    \eqref{eq:theta}. Thus, $\tau = 0$.
  \item
    By virtue of the continuity of $\sigma$, $\eta$ in \eqref{eq:eta} is
    zero. Let $T_1 = \inf\{ t>0: \sigma(t) \le \frac 1 2 \sigma(0)\}$,
    and $T_2 = \inf\{t>0: |b(t)| \le  |b(0)| + 1$. Set $T = T_1\wedge
    T_2$, then $\psi = b/\sigma$ satisfies Nivikov condition
    \eqref{eq:psi}.  Thus,
    $\tau = 0$.
  \end{enumerate}
\end{proof}

We need the following lemma to prove
\thmref{t-vcon}.  The proof of \lemref{l-fox} is
elementary and is thus omitted.

\begin{lem}
  \label{l-fox}
  Let $f: [0,\infty)
  \mapsto
  \mathbb{R}^+$ be a continuous function satisfying
  \begin{equation*}
    \label{eq:l-fox1}
    \lim_{x\to \infty} \frac{f(x)}{x} = 0,
  \end{equation*}
  and denote  $f^*(x) = \max\{f(y): y\le x\}.$
  Then
  \begin{equation*}
    \label{eq:l-fox2}
    \lim_{x\to \infty} \frac{f^*(x)}{x}  = 0.
  \end{equation*}
\end{lem}

\bibliographystyle{plain}
\bibliography{refs}

\end{document}